\documentclass[twocolumn,aps,nofootinbib,superscriptaddress,tightenlines,notitlepage,pra,longbibliography]{revtex4-1}

\usepackage{changes}
\definechangesauthor[name={Antonio}, color=brown]{Ant}

\usepackage[utf8]{inputenc}  
\usepackage[T1]{fontenc}     
\usepackage[english]{babel}  

\usepackage[colorlinks=true,linkcolor=teal,citecolor=teal,urlcolor=teal]{hyperref}

\usepackage{graphicx} 
\usepackage[babel]{microtype}  
\usepackage{amsmath,amssymb,amsthm,bm,amsfonts,mathrsfs,bbm} 
\usepackage{setspace}
\usepackage{times}
\usepackage{braket}
\usepackage[bottom]{footmisc}
\usepackage{blochsphere}
\usepackage{pgfplots}
 
\pgfplotsset{compat=1.17}
\usetikzlibrary{calc,3d,shapes, pgfplots.external, intersections}
\usepackage{algorithm}
\usepackage{physics}
\theoremstyle{plain}
\usepackage[noend]{algorithmic}
\usepackage{eqparbox}

\renewcommand{\abs}[1]{\left\Vert #1\right\Vert}
\usepackage{physics}
\usepackage{bbold}
\usepackage{tikz}
\usepackage{mathtools}
\usetikzlibrary{arrows.meta}
\usepackage{csquotes}

\usepackage{xspace}  
\usepackage{pgfplots}
\usepackage{verbatim}
\usepackage{amsmath}

\newtheorem{theorem}{Theorem}
\newtheorem*{theorem*}{Theorem}

\usepackage{soul}

\newcommand\id{\leavevmode\hbox{\small1\kern-3.3pt\normalsize1}}

\newtheorem{definition}[theorem]{Definition}

\newtheorem{lemma}[theorem]{Lemma}

\newcommand{\commentFre}[1]{{\textcolor{black}{#1}}}

\newcommand{\je}[1]{{\textcolor{black}{#1}}}
\newcommand{\newje}[1]{{\textcolor{black}{#1}}}

\newcommand{\newtext}[1]{{\textcolor{black}{#1}}}

\newcommand{\newnewtext}[1]{{\textcolor{black}{#1}}}

\begin{document}

\title{\newtext{Stochastic} noise can be helpful for variational quantum algorithms}

\author{Junyu Liu}

\affiliation{Pritzker School of Molecular Engineering, The University of Chicago, Chicago, IL 60637, USA}

\affiliation{\newnewtext{School of Computing and Information, University of Pittsburgh, 
Pittsburgh,
PA 15260, USA}}
\affiliation{Chicago Quantum Exchange, Chicago, IL 60637, USA}
\affiliation{Kadanoff Center for Theoretical Physics, The University of Chicago, Chicago, IL 60637, USA}

\author{Frederik Wilde}
\affiliation{Dahlem Center for Complex Quantum Systems,
Freie Universit\"{a}t Berlin, Berlin, 14195, Germany}

\author{Antonio Anna Mele}
\affiliation{Dahlem Center for Complex Quantum Systems,
Freie Universit\"{a}t Berlin, Berlin, 14195, Germany}

\author{Xin Jin}
\affiliation{\newnewtext{School of Computing and Information, University of Pittsburgh, 
Pittsburgh,
PA 15260, USA}}

\affiliation{\newnewtext{Electrical Engineering and Computer Science, Technische Universität Berlin, Berlin, 10587, Germany}}

\author{Liang Jiang}
\affiliation{Pritzker School of Molecular Engineering, The University of Chicago, Chicago, IL 60637, USA}
\affiliation{Chicago Quantum Exchange, Chicago, IL 60637, USA}

\author{Jens Eisert}
\affiliation{Dahlem Center for Complex Quantum Systems,
Freie Universit\"{a}t Berlin, Berlin, 14195, Germany}
\affiliation{Helmholtz-Zentrum Berlin f{\"u}r Materialien und Energie, 14109 Berlin, Germany}
\affiliation{Fraunhofer Heinrich Hertz Institute, 10587 Berlin, Germany}

\begin{abstract}
Saddle points constitute a crucial challenge for first-order gradient descent algorithms. In notions of classical machine learning, they are avoided for example by means of stochastic gradient descent methods. In this work, we provide evidence that the saddle points problem can be naturally avoided in variational quantum algorithms by exploiting the presence of stochasticity. \newtext{We prove convergence guarantees and  present practical examples in numerical simulations and on quantum hardware.} We argue that the natural stochasticity of variational algorithms can be beneficial for avoiding strict saddle points, i.e., those saddle points with at least one negative Hessian eigenvalue. This insight that some levels of shot noise could help is expected to add a new perspective to notions of near-term variational quantum algorithms.
\end{abstract}

\maketitle

Quantum computing has for many years been a hugely inspiring
theoretical idea. Already in the 1980ies it was suggested that
quantum devices could possibly have superior computational
capabilities over computers operating based on classical laws
\cite{Feynman-1986,QuantumTuring}.
It is a relatively recent development that devices have been devised
that may indeed have computational capabilities beyond 
classical means \cite{GoogleSupremacy,Boixo,SupremacyReview,jurcevic_demonstration_2021}. 
These devices are going substantially beyond what was possible not
long ago. And still, they are unavoidably 
noisy and imperfect for many years to come. \newtext{The quantum devices that are available today and presumably will be in the near future are often conceived
as hybrid quantum devices running  
variational quantum algorithms \cite{Cerezo_2021}, 
where a quantum circuit is addressed
by a substantially larger surrounding classical circuit.} This classical
circuit takes measurements from the quantum device and appropriately
varies variational parameters of the quantum device in an update.
Large classes of \emph{variational quantum eigensolvers} (VQE), the \emph{quantum
approximate optimisation algorithm} (QAOA) and models for quantum-assisted machine learning are thought to operate along those lines,
based on suitable \emph{loss functions} to be minimised
\cite{Peruzzo,Kandala,McClean_2016,QAOA,Lukin,bharti_2021_noisy,VariationalReview}. 
\newtext{In fact, 
many near-term quantum algorithms} in the era of 
\emph{noisy intermediate-scale quantum} (NISQ) computing
\cite{preskill_quantum_2018}
\newtext{belong to the class of}
variational quantum algorithms.  
While this is an exciting development, it puts a lot of burden to
understanding how reasonable and practical classical control can be conceived.

\newtext{Generally, 
when the optimisation space is high-dimensional updates of the variational parameters are done via
\emph{gradient 
evaluations} \cite{PhysRevA.99.032331,bergholm2018pennylane,pennylane,Gradients}, while zeroth-order and second-order methods are, in principle, also applicable, but typically only up to a limited number of parameters}. This makes a lot of sense, as one
may think that going downhill in a variational quantum algorithms is a 
good idea. That said, the concomitant classical optimisation 
problems are \newtext{generally not convex optimisation problems and 
the variational landscapes are} marred by 
\je{globally suboptimal}
local optima and saddle points.
\newnewtext{This becomes particularly prominent when the the search space dimension is high, which often leads to most stationary points---points where the gradient vanishes---being saddle points~\cite{bray_statistics_2007}. This is, however, precisely the overparametrized regime where one expects variational quantum algorithms to perform well.}
In fact, it is known that the problems of optimising variational parameters of quantum circuits
\newje{are}
computationally hard in worst case complexity \cite{Bittel}.
\newtext{While this is not of too much concern in practical considerations
(since it is often sufficient to find a ``good'' local minimum instead of the \emph{global} minimum)
and resembles an analogous situation in classical machine learning},
it does point to the fact that one should expect
a rugged optimisation landscape, featuring
different local minima 
as well as saddle points. \newtext{Although in the infinite time limit, the first-order algorithm might avoid saddle points eventually with high probability \cite{lee2016gradient}, but it is shown that in the practical time scale saddle points matter significantly in the general settings of first-order optimization algorithms \cite{du2017gradient}. }

Such saddle points can indeed
be a burden to feasible and practical classical optimisation 
of variational quantum algorithms.

\begin{figure}[h]
\centering
\includegraphics[width=0.46\textwidth]{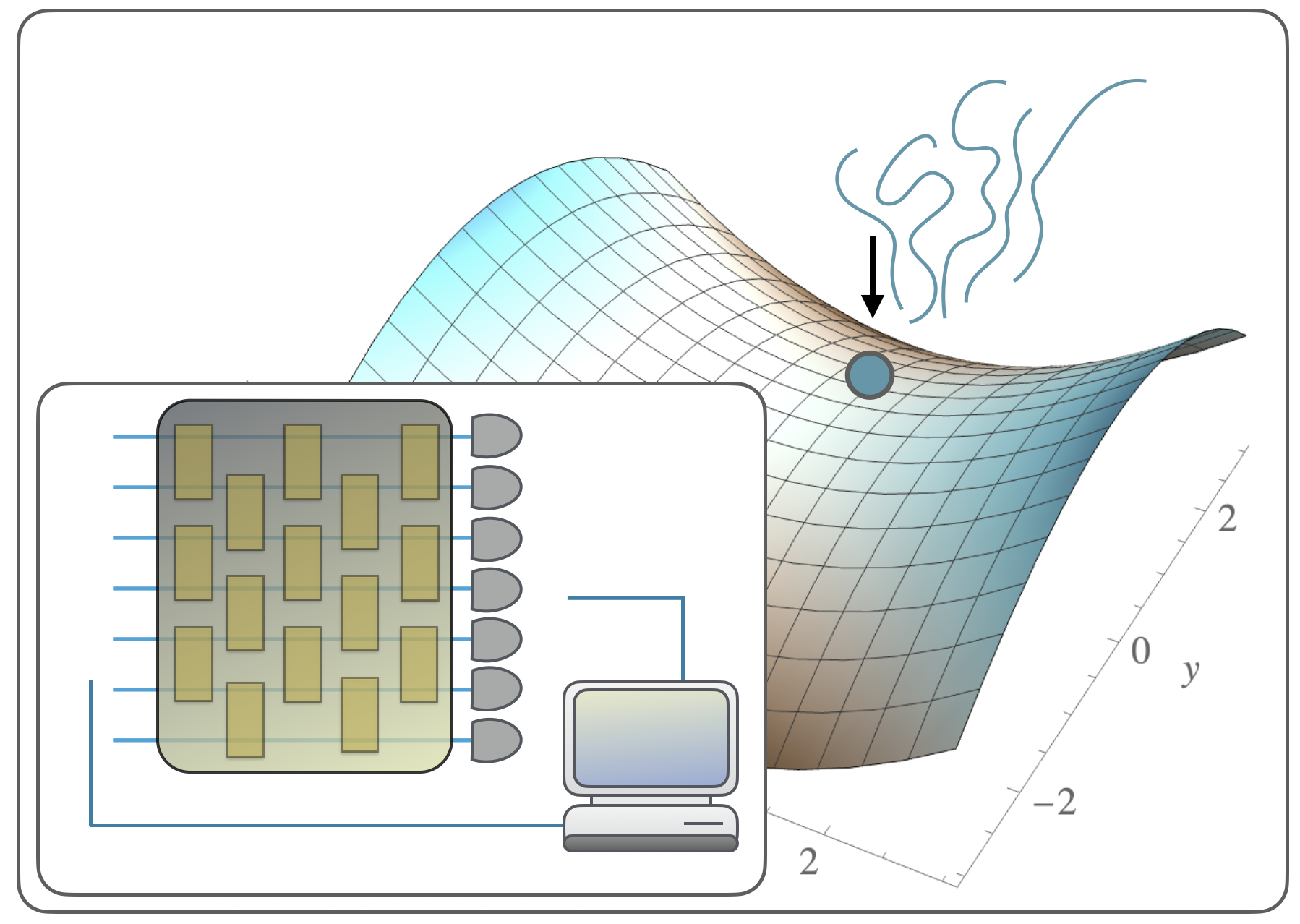}
\caption{Stochasticity in variational quantum algorithms can help in avoiding (strict) saddle points.}
\label{fig:sampling}
\end{figure}
In this work, we establish the notion that 
in such situations, \emph{small noise levels} can actually be of substantial help. 
More precisely, we show that
some levels of \emph{statistical noise} (specifically
the kind of
noise that naturally arises from a finite number of measurements to estimate quantum expectation values) can even be 
beneficial. We get inspiration from and build on a powerful
mathematical theory in \emph{classical machine learning}: 
there, theorems have been established 
that say that ``noise'' can help gradient-descent optimisation not to get stuck at saddle points \cite{JordanSaddlePoints,jain2017non}. 
Building on such ideas we show that they can be adapted and developed to be applicable to
the variational quantum algorithms setting.
Then we argue that the ``natural'' statistical noise of a quantum experiment can play the role of the artificial noise that is inserted by hand in classical machine learning algorithms to avoid saddle points. We maintain the 
\je{precise and} rigorous mindset of Ref.~\cite{JordanSaddlePoints}, but show that the findings
have practical importance and can be made concrete use of 
when running variational quantum algorithms on near-term quantum devices.
In previous studies it has been observed anecdotally that small levels of noise can indeed be helpful for improving the optimisation procedure~\cite{duffield_bayesian_2022, Gradients,Jansen,Rosenkranz,Lorenz, Yelin2, gu_adaptive_2021}.
What is more, 
variational algorithms have been seen as being noise-robust in a
sense \cite{PhysRevA.102.052414}.
\newtext{That said,
while in the past rigorous convergence guarantees have been formulated for convex loss functions of VQAs~\cite{Gradients, gu_adaptive_2021}, in this work we focus on the non-convex scenario, where saddle points and local minima are present. Such
a systematic and rigorous analysis of the type we have conducted that explains the origin of the phenomenon of noise facilitating optimisation has been lacking.}

It is important to stress that for noise we refer to \je{in our theorems is} the type of noise that adds stochasticity to the gradient estimations, such as the use of a finite number of measurements or the zero-average fluctuations that are involved in real experiments. \newtext{ Also, instances of global depolarising noise
are covered as discussed in Appendix 2}.
Thus, in this case,
noise does not mean the generic quantum noise that results from the interaction with the environment characterised by 
\je{\emph{completely positive and trace-preserving}} (CPTP)
maps, which can be substantially detrimental to the 
\newje{performance of the} algorithm \cite{DePalmaVQAs, Stilck_Fran_a_2021}. \newtext{In addition, it has been shown that noisy CPTP maps in the circuit may significantly worsen 
the problem of \emph{barren plateaus} \cite{NoisyBP_Wang_2021,McClean_2018} which is one of the main obstacles to the scalability of \emph{variational quantum algorithms} (VQAs)}.

We perform numerical experiments, and we show examples where optimisations with gradient descent without noise get stuck at saddle points, whereas if we add some noise, we can escape this problem and get to the minimum---convincingly demonstrating the functioning of the approach. We verify the latter not only in a numerical simulation, but also making use of data of a real IBM quantum machine.

\subsection*{Preliminaries}

\label{sec:pre}
In our work, we will show how a class of saddle points, the so-called strict saddle points, can be avoided in noisy gradient descent. In developing our machinery, we build strongly on the 
rigorous results laid out in \je{Ref.}~\cite{JordanSaddlePoints}
\je{and uplift them to the quantum setting at hand}.
\newtext{For this, we do method development in its own right.}
First, we introduce some useful definitions and theorems
(\je{see Ref.}~\cite{JordanSaddlePoints} for a more 
in-depth discussion).

Throughout this work, we consider the problem of minimising a function $\mathcal{L}: \mathbb{R}^p \to \mathbb{R}$.
We indicate its gradient 
at $\theta$ as $\partial \mathcal{L} (\theta)$ and its Hessian matrix at point $\theta$
as $\partial^2 \mathcal{L} (\theta)$. 
\newtext{We denote as $\abs{\cdot}_2$ the $l_2$ norm of a vector. $\abs{\cdot}_{HS}$ and $\abs{\cdot}_\infty$} denote respectively the Hilbert-Schmidt norm and the largest eigenvalue norm of a matrix. We denote as $\lambda_{\min}(\cdot)$ the minimum eigenvalue of a matrix.

\begin{definition}[$L$-Lipschitz function]
A function 
$g: \mathbb{R}^p \to \mathbb{R}^d$ is $L$-Lipschitz if and only if
\begin{align}
    \abs{g (\theta)- g (\phi)}_2\le L\abs{\theta - \phi}_2
\end{align}
for every $\theta$ and $\phi$. 
\end{definition}

\begin{definition}[$\beta$-strong smoothness]
\label{def:Lip}
A differentiable function $\mathcal{L}: \mathbb{R}^p \to \mathbb{R}$ is called $\beta$-strongly smooth if and only if its gradient is a $\beta$-Lipschitz function, i.e.,
\begin{align}
    \abs{\partial \mathcal{L} (\theta)-\partial \mathcal{L} (\phi)}_2\le \beta\abs{\theta - \phi}_2~,
\end{align}
for every $\theta$ and $\phi$. 
\end{definition}

\begin{definition}[Stationary point]
    If $\mathcal{L}$ is differentiable, ${\theta^*}$ is defined a stationary point 
    if
    \begin{equation}
        \abs{\partial \mathcal{L}({\theta^*})}_2=0.
    \end{equation}
\end{definition}

\begin{definition}[$\epsilon$-approximate stationary point]
    If $\mathcal{L}$ is 
    differentiable, ${\theta^*}$ is defined an $\epsilon$-approximate stationary point if
    \begin{equation}
        \abs{\partial \mathcal{L}({\theta^*})}_2\le \epsilon.
    \end{equation}
\end{definition}

\begin{definition}[Local minimum, local maximum, and saddle point]
    If $\mathcal{L}$ is differentiable, a stationary point ${\theta^*}$ is a
    \begin{itemize}
        \item \emph{local minimum}, if there 
        exists $\delta>0$ such that $\mathcal{L}({\theta^*}) \leq \mathcal{L}(\theta)$ for any $\theta$ with $\abs{\theta-{\theta^*}}_2 \leq \delta$.
        \item \emph{Local maximum}, 
        if there exists $\delta>0$ such that $\mathcal{L}({\theta^*}) \geq \mathcal{L}(\theta)$ for any $\theta$ with $\abs{\theta-{\theta^*}}_2 \leq \delta$.
        \item \emph{Saddle point}, otherwise.
    \end{itemize}

\end{definition}

\begin{definition}[$\rho$-Lipschitz Hessian]
\label{def:LipHes}
A twice differentiable function $\mathcal{L}$ has $\rho$-Lipschitz Hessian matrix $\partial^2 \mathcal{L}$ if and only if
\begin{align}
    \abs{\partial^2 \mathcal{L} (\theta)-\partial^2 \mathcal{L} (\phi)}_{\operatorname{HS}}\le \rho\abs{\theta - \phi}_2
\end{align}
for every $\theta$ and $\phi$ (where $\abs{\cdot}_{\operatorname{HS}}$ is the Hilbert-Schmidt norm). 
\end{definition}

\begin{definition}[Gradient descent]
    Given a differentiable function $\mathcal{L}$, the gradient descent algorithm is defined by the update rule
    \begin{align}
\theta_{i}^{t+1}=\theta_{i}^{t}-\eta {{\partial }_{i}}\mathcal{L} (\theta^t)\je{,}
\end{align}
where $\eta>0$ is called \emph{learning rate}.
\end{definition}
The convergence time of the gradient descent algorithm is given by the following theorem \cite{JordanSaddlePoints}.

\begin{theorem}[Gradient descent complexity]
Given a $\beta$-strongly smooth function $\mathcal{L}(\cdot)$, for any $\epsilon>0$, if we set the learning rate as $\eta=1 / \beta$, then the number of iterations required by the gradient descent algorithm such that it will visit an $\epsilon$-approximate stationary point is
$$
\mathcal{O}\left({\frac{\beta\left(\mathcal{L}\left(\mathbf{\theta}_{0}\right)-\mathcal{L}^{\star}\right)}{\epsilon^{2}}}\right),
$$
where $\mathbf{\theta}_{0}$ is the initial point and $\mathcal{L}^\star$ is the value of $\mathcal{L}$ computed in the global minimum.
\end{theorem}

It is important to note that this result does not depend on the number of free parameters. Also, the stationary point at which the algorithm will converge is not necessarily a local minimum, but can also be a saddle point.
Note that a generic saddle point satisfies $\lambda_{\min}(\partial^2 \mathcal{L} (\theta_s))\le0$ where $\lambda_{\min}(\cdot)$ is the minimum eigenvalue. Now we define a subclass of saddle points.
\begin{definition}[Strict saddle point]
$\theta_s$ is a 
\emph{strict saddle} point for a twice differentiable function $\mathcal{L}$ if and only if $\theta_s$ is a stationary point and if the minimum eigenvalue of the Hessian is $\lambda_{\min}(\partial^2 \mathcal{L} (\theta_s))<0$.
\end{definition}
Adding the \emph{strict} condition, we remove the case in which a saddle point satisfies $\lambda_{\min}(\partial^2 \mathcal{L} (\theta_s))=0$.
\newtext{Moreover, note that a local maximum respects our definition of strict saddle point.} 
Analogously to Ref.~\cite{JordanSaddlePoints}, in this 
\je{work,}
we focus on avoiding strict saddle points. Hence, it is useful to introduce the following definition.

\begin{definition}[Second-order stationary point]
    Given a twice differentiable function $\mathcal{L}(\cdot), \mathbf{\theta^*}$ is a second-order stationary point if and only if
\begin{equation}
\partial \mathcal{L}(\mathbf{\theta^*})=\mathbf{0}, \quad \text { and } \quad \lambda_{\min}(\partial^{2} \mathcal{L}(\mathbf{\theta^*})) \ge 0.
\end{equation}
\end{definition}

\begin{definition}[$\epsilon$-second-order stationary point]
    For a $\rho$-Hessian Lipschitz function $\mathcal{L}(\cdot), \mathbf{\theta^*}$ is an $\epsilon$-second-order stationary point if
\begin{equation}
\norm{\partial \mathcal{L}(\mathbf{\theta^*})}_2 \leq \epsilon \quad \text { and } \quad \lambda_{\min}(\partial^{2} \mathcal{L}(\mathbf{\theta^*})) \ge -\sqrt{\rho \epsilon} .
\end{equation}
\end{definition}

Gradient descent makes a non-zero step only when the gradient is non-zero, and thus in the non-convex setting it will be stuck at saddle points. 
A simple variant of GD is the \emph{perturbed gradient descent} (PGD)
method \cite{JordanSaddlePoints} which adds randomness to the iterates at each step.

\begin{definition}[Perturbed gradient descent]
    Given a differentiable function $\mathcal{L} : \mathbb{R}^p \to \mathbb{R}$, the perturbed gradient descent algorithm is defined by the update rule
    \begin{align}
\theta_{i}^{t+1}=\theta_{i}^{t}-\eta \left({{\partial }_{i}}\mathcal{L}\left( { {\theta}^{t}} \right)+{{\zeta }^{t}}\right)~\je{,}
\end{align}
where $\eta>0$ is the learning rate and ${{\zeta }^{t}}$ is a normally distributed random variable with mean $\mu=0$ and variance $\sigma^{2}={r^{2}}/{p}$ with $r \in \mathbb{R}$.
\label{PerturbGF}
\end{definition}
In Ref.~\cite{JordanSaddlePoints}, 
\je{the authors} show that if we pick $r=\tilde{\Theta}(\epsilon)$, PGD will find an $\epsilon$-second-order stationary point in a number of iterations that has only a poly-logarithmic dependence on the number of free parameters, i.e.\je{,} 
it has the same complexity of (standard) gradient descent up to poly-logarithmic dependence.

\begin{theorem}[\cite{JordanSaddlePoints}]
Let the function $\mathcal{L}: \mathbb{R}^d \to \mathbb{R}$ be $\beta$ strongly smooth and such that it has a $\rho$ Lipschitz-Hessian. Then, for any $\epsilon, \delta>0$, the $P G D$ algorithm starting at the point $\theta_{0}$, with parameters $\eta=\tilde{\Theta}(1 / \beta)$ and $r=\tilde{\Theta}(\epsilon)$, will visit an $\epsilon-$second-order stationary point at least once in the following number of iterations, with probability at least $1-\delta$
$$
\tilde{\mathcal{O}}\left(\frac{\beta\left(\mathcal{L}\left(\theta_{0}\right)-\mathcal{L}^{\star}\right)}{\epsilon^{2}}\right)
$$
where $\tilde{\mathcal{O}}$ and $\tilde{\Theta}$ hide poly-logarithmic factors in $p, \beta, \rho, 1 / \epsilon, 1 / \delta$ and $\Delta_{\mathcal{L}}:=\mathcal{L}\left(\theta_{0}\right)-\mathcal{L}^{\star}$. $\mathbf{\theta}_{0}$ is the initial point and $\mathcal{L}^\star$ is the value of $\mathcal{L}$ computed in the global minimum.
\label{GaussianPGD}
\end{theorem}
This theorem \je{has been} proven in \je{Ref.~}\cite{JordanSaddlePoints} for Gaussian distributions, but the authors have pointed out that this is not strictly necessary and that it can be generalized to other types of probability distributions in which appropriate concentration inequalities can be applied \newtext{(for a more in-depth discussion, see 
Ref.~\cite{JordanSaddlePoints})}. 

In Ref.~\cite{ExpTimeSaddle}, 
it has 
been shown 
that although the standard GD (without perturbations) almost always escapes the saddle points asymptotically \cite{Lee2017}, there are (non-pathological) cases in which the optimisation requires exponential time to escape. 
This highlights the importance of using gradient descent with perturbations.

 \begin{figure}[h]
\centering
\includegraphics[width=0.46\textwidth]{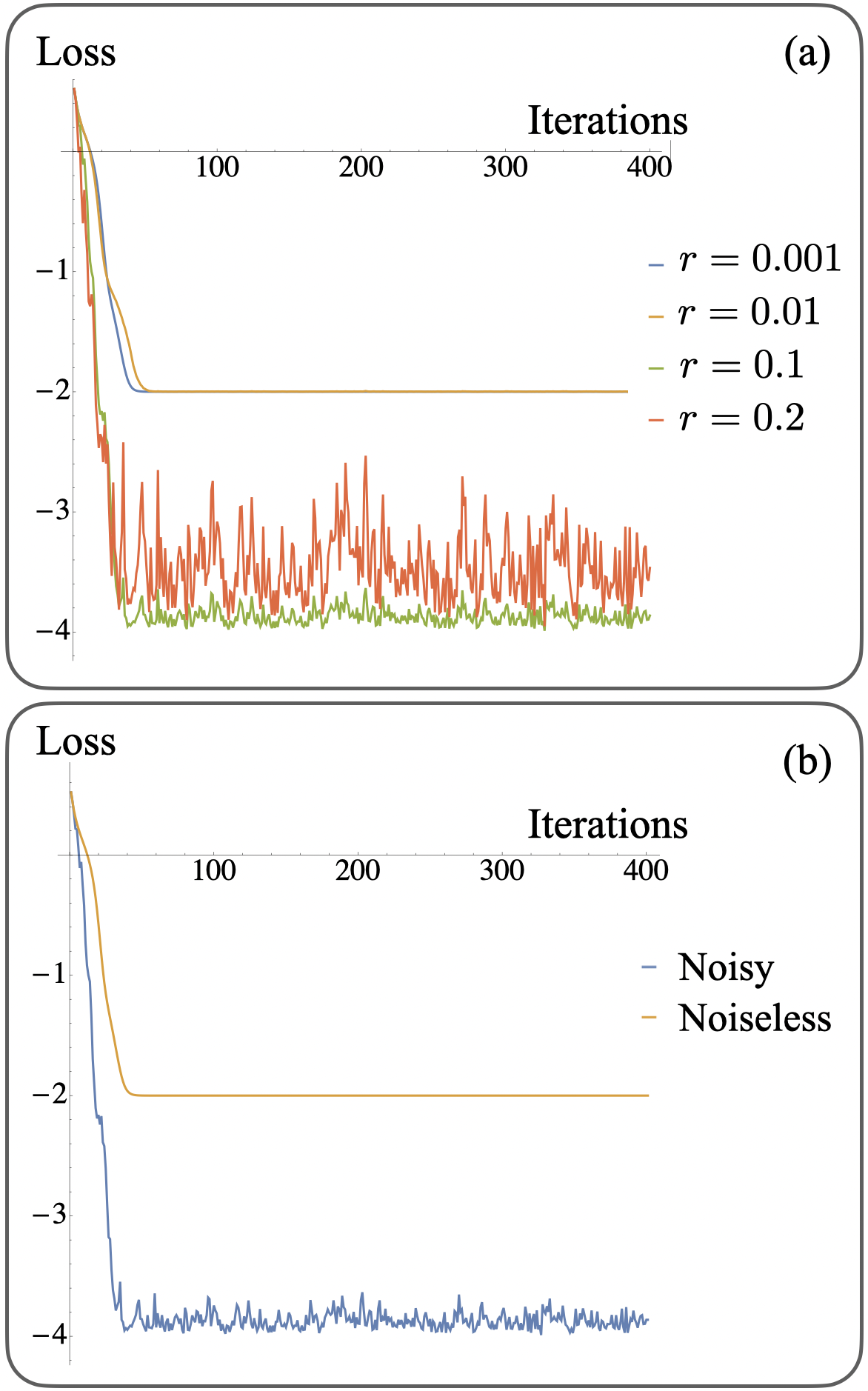}
\caption{\newtext{Comparison of the loss evolution with or without noise. The noise levels are manually-added Gaussian distributions, and we keep the same initial conditions. (a) Four different values of the {standard deviation $r$}. (b) Noiseless case and the noisy case with the standard deviation of the noise $r=0.1$.} 
}
\label{fig:fake}
\end{figure}

\subsection*{Statistical noise in variational quantum algorithms}

Our analysis focuses on variational quantum algorithms in which the loss function to be minimised has the following form
\begin{align}
    \label{eq:loss}
    \newtext{\mathcal{L}\left(\theta\right):=\bra{0} U^{\dagger}(\theta)OU(\theta )\ket{0}}
\end{align} where $O$ is a Hermitian operator and $U(\theta)$ is a parameterised unitary of the form 
\begin{align}
    \newtext{U({\theta}):=\prod_{\ell=1}^{p} W_{\ell} \exp \left(i \theta_{\ell} X_{\ell}\right)},
\end{align}
\newtext{where $W_{\ell}$ and $X_{\ell}$ are,  respectively,  fixed unitaries and Hermitian operators.}
Theorem \ref{GaussianPGD} above assumes that the loss function to minimise is $\beta$-strongly smooth and has a $\rho$-Lipschitz Hessian. \newtext{To guarantee that these conditions are met for loss functions of parametrised 
quantum circuits, we provide the following theorem.}

\begin{theorem}[\newnewtext{Conditions for loss functions of parametrised quantum circuits}]
\label{thm:smoothness}
\newtext{The loss function given in Eq.~(\ref{eq:loss}) with $\theta$ being a $p$-dimensional vector is $\beta$-strongly smooth and has a $\rho$-Lipschitz Hessian.
In particular, we have
\begin{align}
    \label{eq:betabound}
    &\beta \le 2^2p \Vert O\Vert_\infty\max_{j=1,\dots,p}\Vert X_{j}\Vert_\infty^2,\\
    \label{eq:rhobound}
    &\rho \le 2^3p^\frac{3}{2} \Vert O\Vert_\infty\max_{j=1,\dots,p}\Vert X_{j}\Vert_\infty^3.
\end{align}}
\end{theorem}

\newtext{We provide a detailed proof in Appendix 1.}
\newtext{It is important to observe that for typical VQAs, the observable $O$ associated to the loss function and the components $X_i$ have an operator norm that grows at most polynomially with the number of qubits, so also $\beta$ and $\rho$ will grow at most polynomially.}
\newtext{This is because the circuit depth $p$ must be chosen to be at most $\mathcal{O}\left({\rm poly}(n)\right)$ for $n$ qubits, $X_l$ as well as $O$ are usually chosen to be Pauli strings in which case their operator norms is $1$ or to be linear combinations of $\mathcal{O}\left({\rm poly}(n)\right)$ many Pauli strings with $\mathcal{O}\left({\rm poly}(n)\right)$ coefficients (as in QAOA), therefore, by the triangle inequality, their operator norm is bounded by $\mathcal{O}\left({\rm poly}(n)\right)$. Sometimes $O$ is also chosen to be a quantum state \cite{Cerezo_2021}, therefore with operator norm bounded by $1$}.
\newtext{Hence, the number of iterations in Theorem \ref{GaussianPGD} does not grow exponentially in the number of qubits.}

\newtext{The previous results can be easily generalized for the case of differentiable and bounded loss functions which are functions of expectation 
values, i.e., 
\begin{align}
    \mathcal{L}(\theta )=f\left( \left\langle 0\left| {{U}^{\dagger }}(\theta )OU(\theta ) \right|0 \right\rangle\right).
\end{align}
In fact, we observe that if $\mathcal{L}$ and $g$ are Lipschitz functions, then
\begin{eqnarray}
    \vert \mathcal{L}(g({\theta}))-\mathcal{L}(g( {\theta^\prime}))\vert
    & \le L_\mathcal{L} \Vert g(\theta)-g(\theta^\prime)\Vert_2 \\
    & \le L_\mathcal{L} L_g \Vert \theta- \theta^\prime\Vert_2.
    \nonumber
\end{eqnarray}
In addition, if $\mathcal{L}$ is a differentiable function with bounded derivatives on a convex set, then (because of the mean value theorem) $\mathcal{L}$ is Lipschitz on this set. From this follows that if $\mathcal{L}$ is a differentiable function with bounded derivatives of a quantum expectation value (whose image defines a bounded $\mathbb{R}$ interval), then it is Lipschitz.
Moreover, the sum of Lipschitz functions is a Lipschitz function. Therefore, functions of expectation values commonly used in machine learning tasks, such as the
\emph{mean squared error}, satisfy the Lipschitz condition.}

Moreover\je{,} Theorem \ref{GaussianPGD} assumes that at each step of the gradient descent a normally distributed random variable is added to the gradient, namely $\theta_{i}^{t+1}=\theta_{i}^{t}-\eta\left( {{\partial }_{i}}\mathcal{L}\left( { {\theta}^{t}} \right)+{{\zeta }^{t}}\right)$.
In VQAs the partial derivatives are \je{commonly} estimated using a finite number of measurements, such as by the so-called \emph{parameter shift rule}~\cite{GradientsQcomp}. 
Here, the update rule for the gradient descent is
\begin{equation}
    \theta_{i}^{t+1}=\theta_{i}^{t}-\eta \hat{g}_i\left( { {\theta}^{t}} \right),
\end{equation}
where $\hat{g}_i\left( { {\theta}^{t}} \right)$ is an estimator of the partial derivative ${{\partial }_{i}}\mathcal{L}\left( { {\theta}^{t}} \right)$ obtained by a finite number of measurements $N_{\rm shots}$ from the quantum device. Moreover, we
define 
\begin{equation}
\hat{\zeta}^t_{N_{\rm shots}} := {{\partial }_{i}}\mathcal{L}\left( { {\theta}^{t}} \right)-\hat{g}_i\left( { {\theta}^{t}} \right).
\end{equation}
Note that $\hat{\zeta}^t_{N_{\rm shots}}$ is a random variable with zero expectation value.
Therefore\je{,} we have
\begin{equation}
\theta_{i}^{t+1}=\theta_{i}^{t}-\eta\left( {{\partial }_{i}}\mathcal{L}\left( { {\theta}^{t}} \right)+ \hat{\zeta}^t_{N_{\rm shots}}\right)  \je{.}
\end{equation}
The ``noise'' $\hat{\zeta}^t_{N_{\rm shots}}$ will play the role of the noise that is added by hand in the perturbed-gradient descent Algorithm \ref{PerturbGF}.
However, we cannot control exactly the distribution of such random variable, nor the variance.
\newtext{However, it is to be expected that in the limit of many measurement shots, by the central limit theorem, the noise encountered in practice will be close to the noise considered here, a Gaussian distribution.}

\subsection*{Numerical and quantum experiments}

In this section\je{,} we discuss the results of numerical and quantum experiments we \je{have} performed to show that stochasticity can help escape saddle points. 
Our results suggest that statistical noise  
\newtext{leads to a non-vanishing} probability of not getting stuck in a saddle point and  
\newtext{thereby reaching} a lower value of the loss function.
\je{These numerical experiments also complement the
rigorous results that are proven to be valid
under very precisely defined conditions, while
the intuition developed here is expected to
be more broadly applicable, so that the 
rigorous results can be seen as proxies for a
more general mindset.} 
We \je{have} also observed this phenomena in a real IBM quantum device. \je{We have done so to convincingly stress the significance of our results
in practice.}

Let us first consider the Hamiltonian $\newtext{O}=\sum_{i=1}^{N=4}Z_i$.
The loss function we consider is defined as the expectation value of such a Hamiltonian over the parametrised quantum circuit $\texttt{qml.StronglyEntanglingLayers}$ implemented
in \emph{Pennylane}~\cite{bergholm2018pennylane}, where two layers of the circuit are used.

In all our experiments we first \newtext{initialise the parameters in multiple randomly chosen values. Next, we select the initial points for which the optimisation process gets stuck at a suboptimal loss-function value, thereby focusing on cases in which saddle points constitute a significant problem for the (noiseless) optimiser.
This selection can be justified by the fact that the loss function defined by $O$ is trivial to begin with. The relevant aspect of this experiment is to study situations in which optimiser encounters saddle points.
As such we exclusively investigate these specifically selected initial points by subsequently initialising the noisy optimiser with them.}

As a proof of principle, we first show the results of an exact simulation (i.e., the expectation values are not estimated using a finite number of shots, but are calculated exactly) in which noise is added manually at each step of the gradient descent. The probability distribution associated to the noise is chosen to be a Gaussian distribution with mean $\mu=0$ and variance $\sigma^2=r^2$.
Figure \ref{fig:fake} shows the difference between the noiseless, and noisy calculations with the same initial conditions of the gradient descent, when the noise is from random Gaussian perturbations added manually. 
Figure \ref{fig:performanceGaussian} shows the performance of the experiment, defined as \je{${1}/({\mathcal{L}-\mathcal{L}_\text{opt}})$} 
as a function of the noise parameter. Here, we can find a critical value of noise, leading to saddle-point avoidance. 
Figure \ref{fig:sampling} specifically addresses quantum noise levels, with simulated results about purely statistical noise levels (shot noise) and device noise (simulated 
by making use of the \newtext{noise model of} actual quantum hardware \texttt{IBM Qiskit}). 

\je{It} should be noted that including device noise generally also means dealing with \emph{completely positive
trace preserving} maps that can lead to a different loss function, with new local minima, new saddle points and a flatter landscape \cite{NoisyBP_Wang_2021}. However, even in this case we observe an improvement in performance using the same initial parameters leading to the saddle point in the noiseless case. \je{This is perfectly in line with the intuition developed here, as long as the effective noise emerging can be seen as a small perturbation of the reference circuit featuring a given loss landscape that is then in effect perturbed by stochastic noise.}

\begin{figure}[h]
\centering
\includegraphics[width=0.46\textwidth]{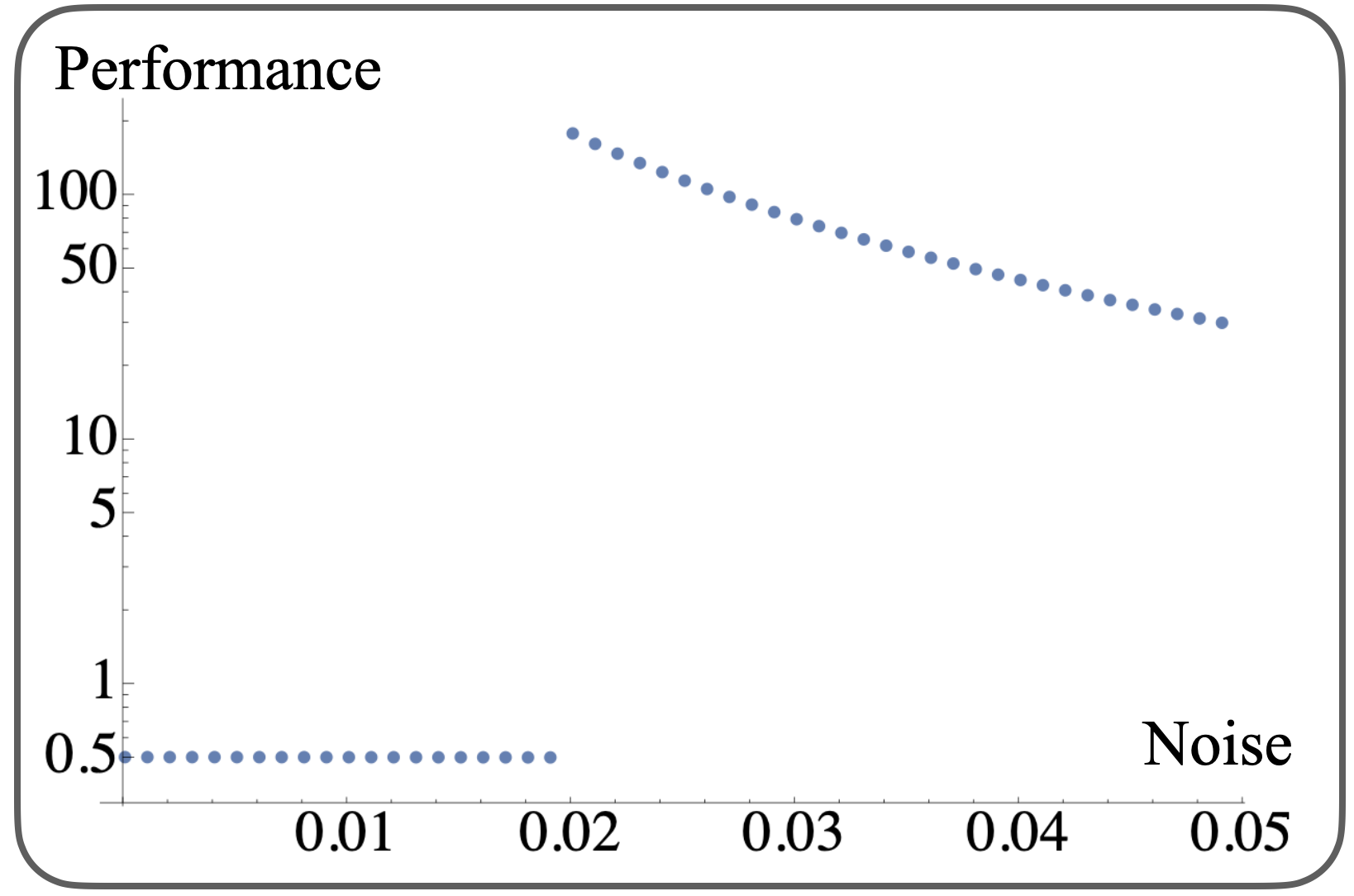}
\caption{We quantify the performance against the size of the noise $r$ (classical Gaussian noise) by ${1}/({\mathcal{L}-\mathcal{L}_\text{opt}})$.} %
\label{fig:performanceGaussian}
\end{figure}

\begin{figure}[h]
\centering
\includegraphics[width=0.47\textwidth]{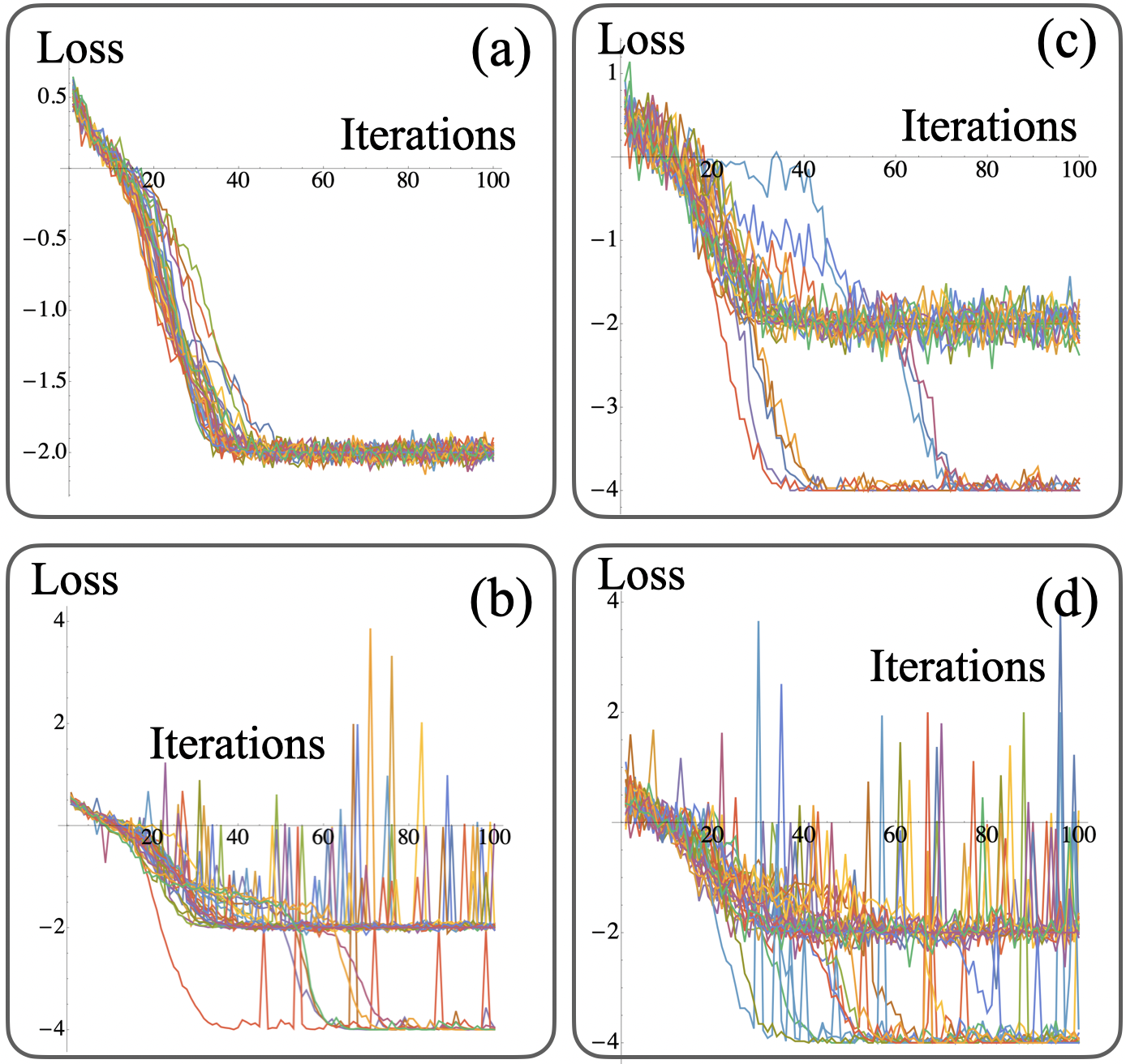}
\caption{\newje{Saddle-point avoidance from quantum noise. We prepare 30 instances starting from the same initial condition. When noise levels are small \newtext{(a), with purely measurement noise, (b), including device noise, and shot number is 1000), most trajectories cannot jump out of the saddle points. When noise levels are larger (c), with purely measurement noise, (d), including device noise, and shot number is 70), we have a probability to jump towards the global minimum.}}}
\label{fig:sampling}
\end{figure}
 
Aside from the quantum machine learning example, we also provide another instance in \emph{variational quantum eigensolvers} (VQEs). Here, we use the Hamiltonian associated to the Hydrogen molecule $\text{H}_2$, which is a 4-qubit Hamiltonian obtained by the fermionic one performing a Jordan-Wigner transformation. We specifically use the same circuit from \texttt{h2.xyz}, the \emph{Hydrogen VQE} example in \emph{Pennylane} \cite{pennylane}. 
Also here, given initial parameters that led to saddle points in the noiseless case, we find that starting by the same parameters and adding noise can lead to saddle-point avoidance. Results are shown in Fig.~\ref{fig:fake_vqe} where we compare the noiseless and noisy simulation. 

\begin{figure}[h]
\centering
\includegraphics[width=0.46\textwidth]{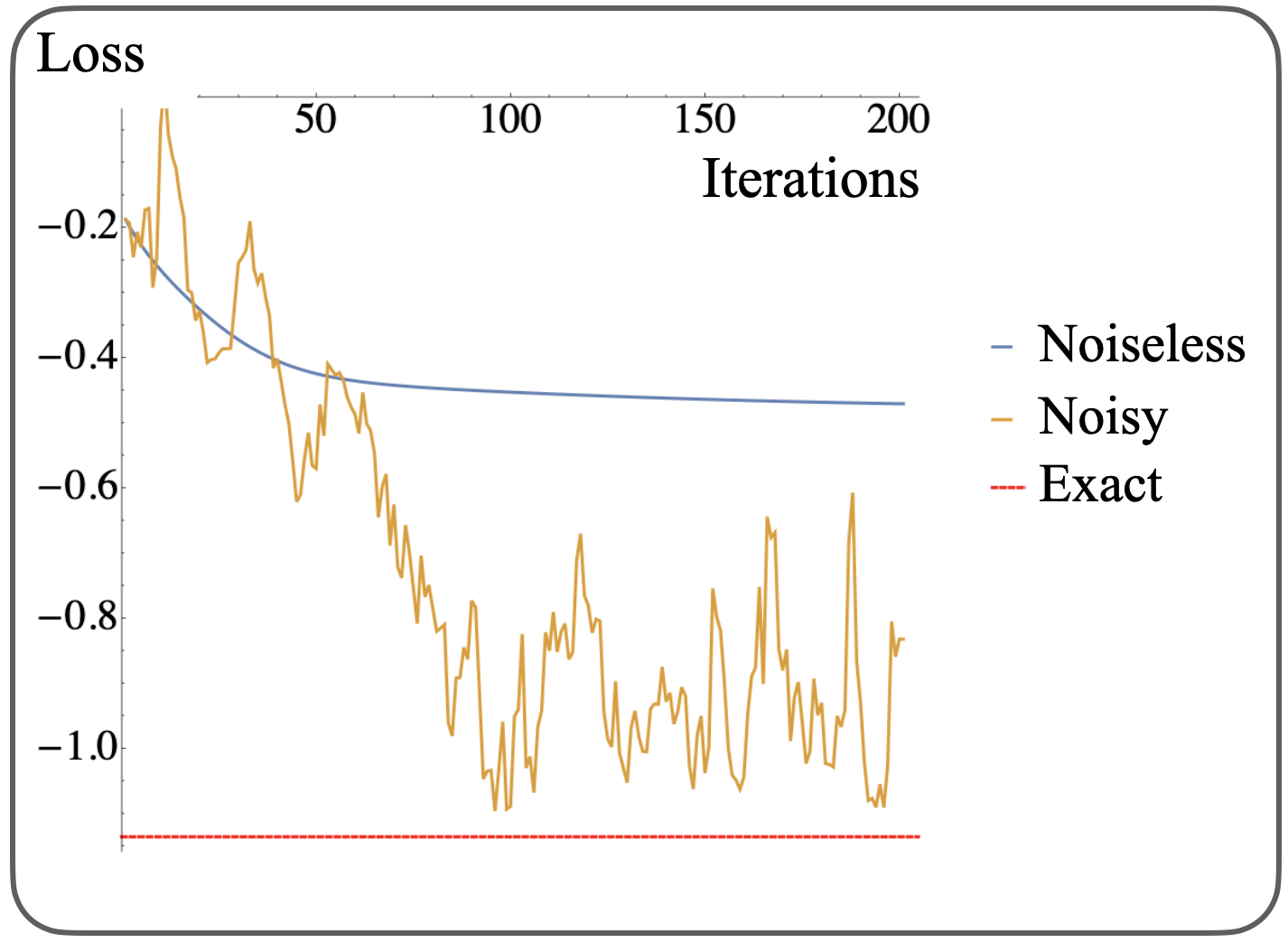}
\caption{Comparison of the loss evolution with or without noise with Hydrogen VQE. The noise is manually drawn from Gaussian distributions with the standard deviation $0.2$, and we keep the same initial conditions. We compare the noiseless case, noisy case and the exact solution.}
\label{fig:fake_vqe}
\end{figure}


\je{To further provide evidence of the functioning of our
approach suggested and the rigorous insights established,
we put the findings into contact of the results of a}
real experiment example in the \texttt{IBM Qiskit} environment. We use the Hamiltonian $\commentFre{O}=\sum_{i=1}^{N=4}Z_i$ that we used in our first numerical simulation with 
\je{four} qubits and \je{two}
layers. We run the experiment using as initial condition the one that lead to a saddle point in the noiseless case. We use the \texttt{IBMQ\_Jakarta} device with 10000 shots. The result in \je{Fig.}~\ref{fig:converge_experiment} shows that it is possible to obtain a lower value of the cost function than that of the simulation without noise that 
\je{has been stuck} in a saddle point. 

\newnewtext{One may also ask at this point 
what the ``sweet spot'' of 
the appropriate stochastic noise might possibly be. It is known that for most local quantum circuits being subject to constant noise levels under general local qubit noise, one can expect the maximum attainable circuit depth to scale like $\mathcal{O}(\log(n))$ \cite{Nonunital}. Also, it is known that Pauli expectation values between two different input states are exponentially suppressed in the circuit depth $d$ \cite{Nonunital}. This implies that the logarithm of the noise levels must be chosen as $\mathcal{O}(1/{\rm poly}(d))$ to be non-detrimental for read-out, but still make sure that one avoids saddle points in variational optimization. Too high noise levels, in the form of contractive device noise, will eventually lead to noise-induced barren plateaus \cite{NoisyBP_Wang_2021} and an eventual disappearance of any distinguishability of outputs.}

\begin{figure}[h]
\centering
\includegraphics[width=0.46\textwidth]{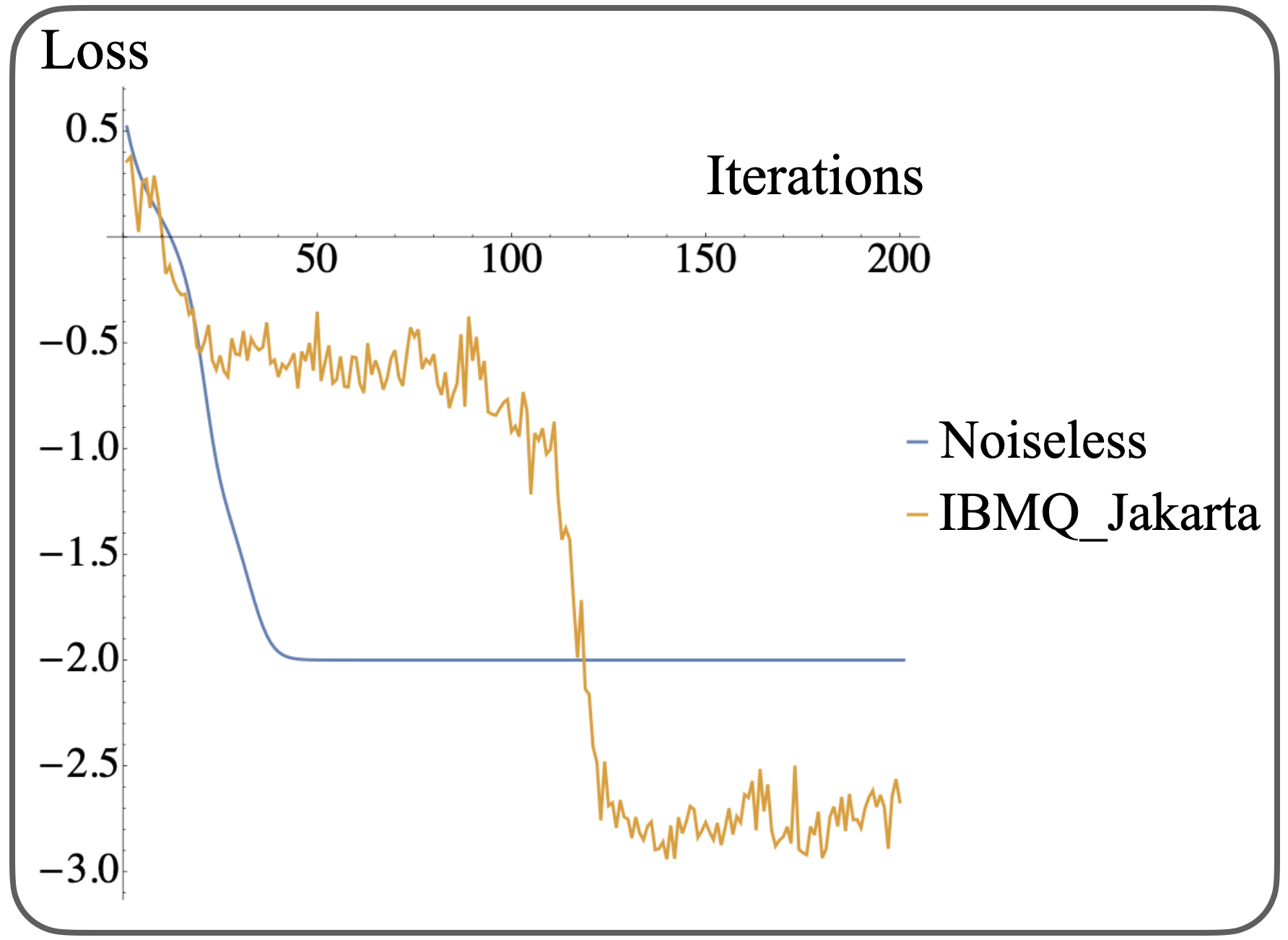}
\caption{A real quantum experiment. We use \texttt{IBMQ\_Jakarta} device with 10000 shots. }
\label{fig:converge_experiment}
\end{figure}

\newnewtext{To investigate the significance of saddle points in the case of variational quantum algorithms, we examine the generality of initial points that could be trapped by saddle points. In Fig.~\ref{fig:moreinitial}, we have studied 1000 initial variational angles randomly sampled between $[0,2\pi)$ in the same setup of Fig.~\ref{fig:performanceGaussian} of the main text with four qubits. Under identical gradient 
descent conditions, we have observed that 305 initial points lead to saddle points rather than local minima in the absence of noise. Consequently, the estimated probability of getting stuck near saddle points in our study is approximately $30.5\%$. Classical non-convex optimization theory demonstrates that saddle 
points are not anomalies but rather common features in loss function landscapes, making stochastic gradient descents crucial for most traditional machine learning applications. While a $30.5\%$ failure rate in optimization is manageable by simply repeating the algorithm multiple times to recover the global minimum with high probability, we anticipate that in quantum machine learning, getting trapped by saddle points will also be a general occurrence. Moreover, the likelihood of encountering saddle points is expected to grow significantly in higher-dimensional parameter spaces~\cite{bray_statistics_2007}. Here is a straightforward 
explanation: Saddle points are determined by the signs of Hessian eigenvalues. In models with $p$ parameters, there are $p$ Hessian eigenvalues. Assuming equal chances of positive and negative eigenvalues, the probability of obtaining a positive semi-definite Hessian becomes exponentially small ($2^{-p}$). Therefore, as the model size increases, saddle points become increasingly prevalent.
}

\begin{figure}[h]
\centering
\includegraphics[width=0.4\textwidth]{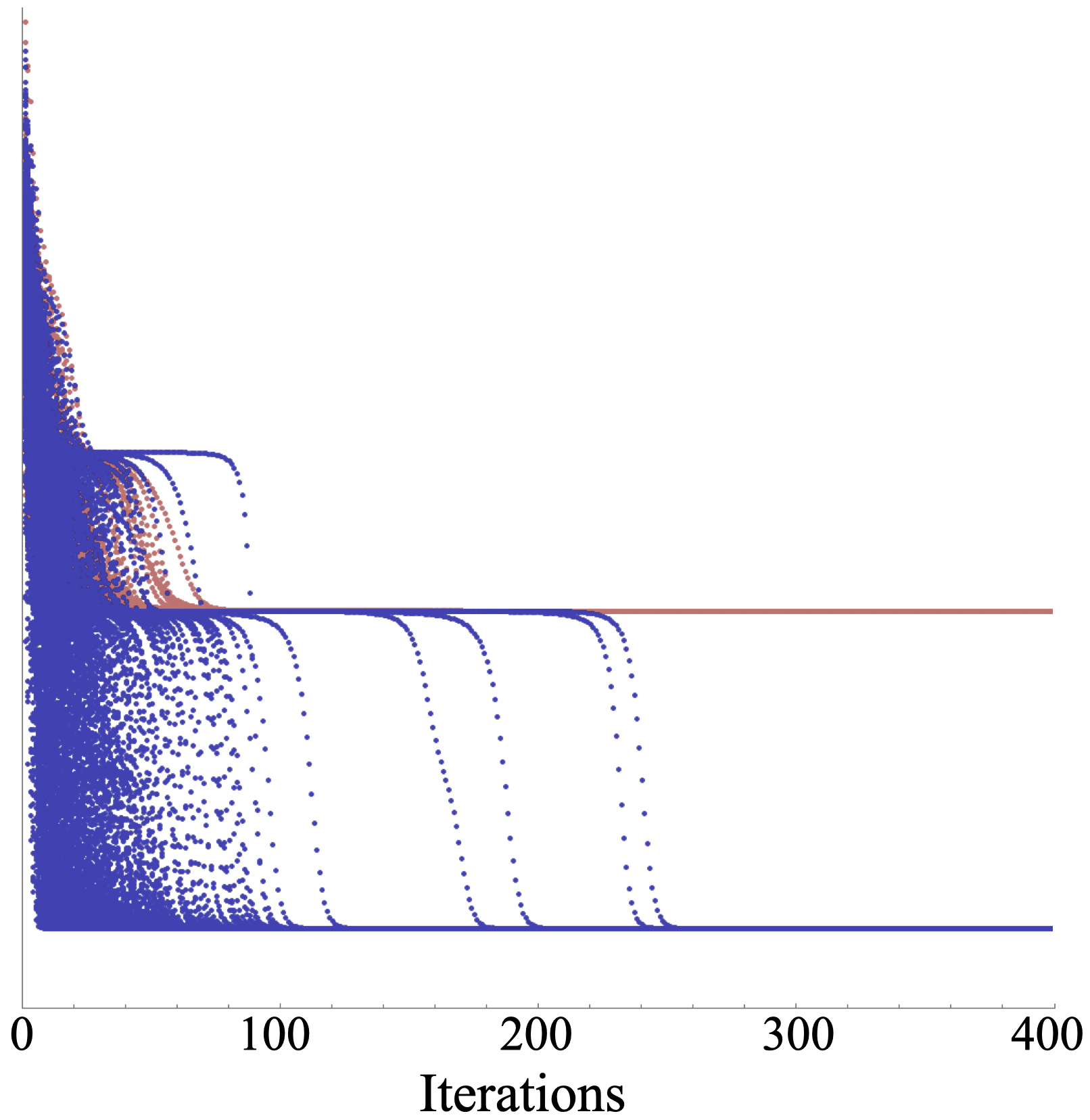}

\caption{{More initial points under the same gradient descent dynamics. We have studied the same gradient descent dynamics as in Fig.~3 of the main text with four qubits, where 1000 initial points have been sampled randomly. {We have found that in 305 cases (labeled in red), the gradient descent gets trapped by saddle points.}} }
\label{fig:moreinitial}
\end{figure}

\subsection*{Conclusion and outlook}
In this work, we have proposed small \je{stochastic noise} levels as an instrument to facilitate variational quantum algorithms. \je{This noise can be substantial, but should not be too large:} The way a noise level can strike the balance in overcoming getting stuck in saddle points and being detrimental is in some ways reminiscent of
the phenomenon of \emph{stochastic resonance} in statistical
physics \cite{StochasticResonance}. \newtext{This is a phenomenon in which
suitable small increases in levels of noise can increase in a metric of the 
quality of signal transmission, resonance or detection performance, 
rather than a decrease. Here, also} fine
tuned noise levels can facilitate resonance behaviour
and avoid getting trapped. 

\newnewtext{It is worth stressing that our results focus specifically on shot noise, which can help in overcoming saddle points. This is fundamentally distinct from contractive noise maps, such as depolarizing noise, which may induce barren plateaus~\cite{NoisyBP_Wang_2021} in variational quantum algorithms and does not help in avoiding saddle points.
It has been shown that shot noise leads to barren plateaus \emph{only} in the case of global observables~\cite{NoisyBP_Wang_2021}, which is not the setting considered in this work.
At the end of the day, one should expect specifics of the noise map, as for overly large noise levels of a certain type, the performance of variational approaches will also worsen \cite{Alexeev}.}

\newnewtext{We also emphasize that our analysis is focused on intrinsic quantum noise in quantum computing (shot noise), not merely the setting where one adds extra classical Gaussian noise to the gradient. While we use Gaussian noise to approximate the shot noise under certain conditions (when the number of measurement shots is large), our primary objective is to show that quantum noise---an unavoidable feature of quantum systems---can enhance optimization within limited noise level ranges. We do not claim that noise is universally beneficial nor advocate for intentionally adding extra classical noise. Instead, we emphasize the importance of identifying an optimal level of inherent quantum noise that balances performance and practicality. This perspective suggests a more tolerant approach to quantum noise in gradients, reducing the reliance on perfectly noise-less quantum systems while mitigating saddle point entrapment.}

On a higher level, 
\newje{our} work invites to think more deeply about 
\newtext{the use of classical stochasic noise in variational 
quantum algorithms as well as ways to prove performance guarantees
about such approaches. For example,}
\emph{Metropolis sampling}
inspired classical algorithms %
\newtext{in which a stochastic process satisfying detailed balance is set
up over variational quantum circuits 
may assist in avoiding}
getting 
stuck in rugged energy landscapes.

\newtext{It is also interesting to note that the technical results obtained here provide further insights into an alternative interpretation of the setting discussed here. Instead of 
regarding the noise as stochastic noise that facilitates the optimisation in the proven fashion established here,
one may argue that the noise channels associated with the
noise alter the variational landscapes in the first place
\cite{PhysRevA.104.022403,PRXQuantum.2.040309}. For example, 
such quantum channels are known to be able to 
break parameter symmetries in over-parameterised variational algorithms. It is plausible to assume that these altered
variational landscapes may be easier to optimise over.
It is an interesting observation in its own right that the 
technical results obtained here also have implications to
this alternative viewpoint, as the convergence
guarantees are independent on the interpretation.}
It is the hope that the present work puts the role of stochasticity in variational quantum computing into a new perspective, and contributes to a line of thought exploring the use of suitable noise and sampling for enhancing quantum computing schemes.

\subsection*{Author contributions}
J.~E.~and J.~L.~have suggested to exploit classical stochastic noise in variational quantum algorithms, and to prove convergence guarantees for the performance of the resulting algorithms. J.~L., A.~A.~M., F.~W., and J.~E.~have proven the theorems of convergence. J.~L.~and F.~W.~have devised and conducted the numerical simulations. J.~L.~has performed quantum device experiments under the guidance of L.~J. X.~J. has attended the discussions and contributed in the scientific updates of the draft. All authors have discussed the results and have written the manuscript.

\subsection*{Competing interests}
There are no competing interests.

\subsection*{Data and code availability}
\newtext{The data and code used for the experiments are available at  \texttt{https://github.com/junyuphybies/saddlepoints/}.} 

\subsection*{Acknowledgments}

J.~L.~is supported in part by the International Business Machines (IBM) Quantum through the Chicago Quantum Exchange, and the Pritzker School of Molecular Engineering at the University of Chicago through AFOSR MURI (FA9550-21-1-0209).
J.~L. and X.~J are supported in part by the University of Pittsburgh, School of Computing and Information, Department of Computer Science, Pitt Cyber, PQI Community Collaboration Awards, and by NASA under award number 80NSSC25M7057. 
F.~W., 
A.~A.~M.~and J.~E.~thank 
the ERC (DebuQC),
the BMBF (Hybrid, MuniQC-Atoms, DAQC),
the BMWK (EniQmA, PlanQK), 
the MATH+ Cluster of Excellence, 
the Quantum Flagship (Millenion, PasQuans2),
the Einstein Foundation (Einstein Unit on 
Quantum Devices),
\newnewtext{Berlin Quantum},
the QuantERA (HQCC), the Munich Quantum 
Valley (K8), the DFG (CRC 183),
\newnewtext{and the European Research Council (DebuQC)}
for support. L.~J.~acknowledges support from the ARO (W911NF-18-1-0020, W911NF-18-1-0212), ARO MURI (W911NF-16-1-0349), AFOSR MURI (FA9550-19-1-0399, FA9550-21-1-0209), DoE Q-NEXT, NSF (EFMA-1640959, OMA-1936118, EEC-1941583), NTT Research, and the Packard Foundation (2013-39273). This research used resources of the Oak Ridge Leadership Computing Facility, which is a DOE Office of Science User Facility supported under Contract DE-AC05-00OR22725.

\bibliography{sampling}

\begin{thebibliography}{50}%
\makeatletter
\providecommand \@ifxundefined [1]{%
 \@ifx{#1\undefined}
}%
\providecommand \@ifnum [1]{%
 \ifnum #1\expandafter \@firstoftwo
 \else \expandafter \@secondoftwo
 \fi
}%
\providecommand \@ifx [1]{%
 \ifx #1\expandafter \@firstoftwo
 \else \expandafter \@secondoftwo
 \fi
}%
\providecommand \natexlab [1]{#1}%
\providecommand \enquote  [1]{``#1''}%
\providecommand \bibnamefont  [1]{#1}%
\providecommand \bibfnamefont [1]{#1}%
\providecommand \citenamefont [1]{#1}%
\providecommand \href@noop [0]{\@secondoftwo}%
\providecommand \href [0]{\begingroup \@sanitize@url \@href}%
\providecommand \@href[1]{\@@startlink{#1}\@@href}%
\providecommand \@@href[1]{\endgroup#1\@@endlink}%
\providecommand \@sanitize@url [0]{\catcode `\\12\catcode `\$12\catcode
  `\&12\catcode `\#12\catcode `\^12\catcode `\_12\catcode `\%12\relax}%
\providecommand \@@startlink[1]{}%
\providecommand \@@endlink[0]{}%
\providecommand \url  [0]{\begingroup\@sanitize@url \@url }%
\providecommand \@url [1]{\endgroup\@href {#1}{\urlprefix }}%
\providecommand \urlprefix  [0]{URL }%
\providecommand \Eprint [0]{\href }%
\providecommand \doibase [0]{http://dx.doi.org/}%
\providecommand \selectlanguage [0]{\@gobble}%
\providecommand \bibinfo  [0]{\@secondoftwo}%
\providecommand \bibfield  [0]{\@secondoftwo}%
\providecommand \translation [1]{[#1]}%
\providecommand \BibitemOpen [0]{}%
\providecommand \bibitemStop [0]{}%
\providecommand \bibitemNoStop [0]{.\EOS\space}%
\providecommand \EOS [0]{\spacefactor3000\relax}%
\providecommand \BibitemShut  [1]{\csname bibitem#1\endcsname}%
\let\auto@bib@innerbib\@empty
\bibitem [{\citenamefont {Feynman}(1986)}]{Feynman-1986}%
  \BibitemOpen
  \bibfield  {author} {\bibinfo {author} {\bibfnamefont {R.~P.}\ \bibnamefont
  {Feynman}},\ }\bibfield  {title} {\enquote {\bibinfo {title} {Quantum
  mechanical computers},}\ }\href {\doibase 10.1007/BF01886518} {\bibfield
  {journal} {\bibinfo  {journal} {Found. Phys.}\ }\textbf {\bibinfo {volume}
  {16}},\ \bibinfo {pages} {507--531} (\bibinfo {year} {1986})}\BibitemShut
  {NoStop}%
\bibitem [{\citenamefont {Deutsch}(1985)}]{QuantumTuring}%
  \BibitemOpen
  \bibfield  {author} {\bibinfo {author} {\bibfnamefont {D.}~\bibnamefont
  {Deutsch}},\ }\bibfield  {title} {\enquote {\bibinfo {title} {{Quantum
  theory, the Church-Turing principle and the universal quantum computer}},}\
  }\href {\doibase 10.1098/rspa.1985.0070} {\bibfield  {journal} {\bibinfo
  {journal} {Proc. Roy. Soc. A}\ }\textbf {\bibinfo {volume} {400}},\ \bibinfo
  {pages} {97--117} (\bibinfo {year} {1985})}\BibitemShut {NoStop}%
\bibitem [{\citenamefont {Arute}\ \emph {et~al.}(2019)\citenamefont {Arute}
  \emph {et~al.}}]{GoogleSupremacy}%
  \BibitemOpen
  \bibfield  {author} {\bibinfo {author} {\bibfnamefont {F.}~\bibnamefont
  {Arute}} \emph {et~al.},\ }\bibfield  {title} {\enquote {\bibinfo {title}
  {Quantum supremacy using a programmable superconducting processor},}\ }\href
  {\doibase 10.1038/s41586-019-1666-5} {\bibfield  {journal} {\bibinfo
  {journal} {Nature}\ }\textbf {\bibinfo {volume} {574}},\ \bibinfo {pages}
  {505--510} (\bibinfo {year} {2019})}\BibitemShut {NoStop}%
\bibitem [{\citenamefont {Boixo}\ \emph {et~al.}(2018)\citenamefont {Boixo},
  \citenamefont {Isakov}, \citenamefont {Smelyanskiy}, \citenamefont {Babbush},
  \citenamefont {Ding}, \citenamefont {Jiang}, \citenamefont {Bremner},
  \citenamefont {Martinis},\ and\ \citenamefont {Neven}}]{Boixo}%
  \BibitemOpen
  \bibfield  {author} {\bibinfo {author} {\bibfnamefont {S.}~\bibnamefont
  {Boixo}}, \bibinfo {author} {\bibfnamefont {S.~V.}\ \bibnamefont {Isakov}},
  \bibinfo {author} {\bibfnamefont {V.~N.}\ \bibnamefont {Smelyanskiy}},
  \bibinfo {author} {\bibfnamefont {R.}~\bibnamefont {Babbush}}, \bibinfo
  {author} {\bibfnamefont {N.}~\bibnamefont {Ding}}, \bibinfo {author}
  {\bibfnamefont {Z.}~\bibnamefont {Jiang}}, \bibinfo {author} {\bibfnamefont
  {M.~J.}\ \bibnamefont {Bremner}}, \bibinfo {author} {\bibfnamefont {J.~M.}\
  \bibnamefont {Martinis}}, \ and\ \bibinfo {author} {\bibfnamefont
  {H.}~\bibnamefont {Neven}},\ }\bibfield  {title} {\enquote {\bibinfo {title}
  {Characterizing quantum supremacy in near-term devices},}\ }\href {\doibase
  10.1038/s41567-018-0124-x} {\bibfield  {journal} {\bibinfo  {journal} {Nature
  Phys.}\ }\textbf {\bibinfo {volume} {14}},\ \bibinfo {pages} {595--600}
  (\bibinfo {year} {2018})}\BibitemShut {NoStop}%
\bibitem [{\citenamefont {Hangleiter}\ and\ \citenamefont
  {Eisert}(2023)}]{SupremacyReview}%
  \BibitemOpen
  \bibfield  {author} {\bibinfo {author} {\bibfnamefont {D.}~\bibnamefont
  {Hangleiter}}\ and\ \bibinfo {author} {\bibfnamefont {J.}~\bibnamefont
  {Eisert}},\ }\bibfield  {title} {\enquote {\bibinfo {title} {Computational
  advantage of quantum random sampling},}\ }\href {\doibase
  10.1103/RevModPhys.95.035001} {\bibfield  {journal} {\bibinfo  {journal}
  {Rev. Mod. Phys.}\ }\textbf {\bibinfo {volume} {95}},\ \bibinfo {pages}
  {035001} (\bibinfo {year} {2023})}\BibitemShut {NoStop}%
\bibitem [{\citenamefont {Jurcevic}\ \emph {et~al.}(2021)\citenamefont
  {Jurcevic}, \citenamefont {{Javadi-Abhari}}, \citenamefont {Bishop},
  \citenamefont {Lauer}, \citenamefont {Bogorin}, \citenamefont {Brink},
  \citenamefont {Capelluto}, \citenamefont {G{\"u}nl{\"u}k}, \citenamefont
  {Itoko}, \citenamefont {Kanazawa}, \citenamefont {Kandala}, \citenamefont
  {Keefe}, \citenamefont {Krsulich}, \citenamefont {Landers}, \citenamefont
  {Lewandowski}, \citenamefont {McClure}, \citenamefont {Nannicini},
  \citenamefont {Narasgond}, \citenamefont {Nayfeh}, \citenamefont {Pritchett},
  \citenamefont {Rothwell}, \citenamefont {Srinivasan}, \citenamefont
  {Sundaresan}, \citenamefont {Wang}, \citenamefont {Wei}, \citenamefont
  {Wood}, \citenamefont {Yau}, \citenamefont {Zhang}, \citenamefont {Dial},
  \citenamefont {Chow},\ and\ \citenamefont
  {Gambetta}}]{jurcevic_demonstration_2021}%
  \BibitemOpen
  \bibfield  {author} {\bibinfo {author} {\bibfnamefont {P.}~\bibnamefont
  {Jurcevic}}, \bibinfo {author} {\bibfnamefont {A.}~\bibnamefont
  {{Javadi-Abhari}}}, \bibinfo {author} {\bibfnamefont {L.~S.}\ \bibnamefont
  {Bishop}}, \bibinfo {author} {\bibfnamefont {I.}~\bibnamefont {Lauer}},
  \bibinfo {author} {\bibfnamefont {D.~F.}\ \bibnamefont {Bogorin}}, \bibinfo
  {author} {\bibfnamefont {M.}~\bibnamefont {Brink}}, \bibinfo {author}
  {\bibfnamefont {L.}~\bibnamefont {Capelluto}}, \bibinfo {author}
  {\bibfnamefont {O.}~\bibnamefont {G{\"u}nl{\"u}k}}, \bibinfo {author}
  {\bibfnamefont {T.}~\bibnamefont {Itoko}}, \bibinfo {author} {\bibfnamefont
  {N.}~\bibnamefont {Kanazawa}}, \bibinfo {author} {\bibfnamefont
  {A.}~\bibnamefont {Kandala}}, \bibinfo {author} {\bibfnamefont {G.~A.}\
  \bibnamefont {Keefe}}, \bibinfo {author} {\bibfnamefont {K.}~\bibnamefont
  {Krsulich}}, \bibinfo {author} {\bibfnamefont {W.}~\bibnamefont {Landers}},
  \bibinfo {author} {\bibfnamefont {E.~P.}\ \bibnamefont {Lewandowski}},
  \bibinfo {author} {\bibfnamefont {D.~T.}\ \bibnamefont {McClure}}, \bibinfo
  {author} {\bibfnamefont {G.}~\bibnamefont {Nannicini}}, \bibinfo {author}
  {\bibfnamefont {A.}~\bibnamefont {Narasgond}}, \bibinfo {author}
  {\bibfnamefont {H.~M.}\ \bibnamefont {Nayfeh}}, \bibinfo {author}
  {\bibfnamefont {E.}~\bibnamefont {Pritchett}}, \bibinfo {author}
  {\bibfnamefont {M.~B.}\ \bibnamefont {Rothwell}}, \bibinfo {author}
  {\bibfnamefont {S.}~\bibnamefont {Srinivasan}}, \bibinfo {author}
  {\bibfnamefont {N.}~\bibnamefont {Sundaresan}}, \bibinfo {author}
  {\bibfnamefont {C.}~\bibnamefont {Wang}}, \bibinfo {author} {\bibfnamefont
  {K.~X.}\ \bibnamefont {Wei}}, \bibinfo {author} {\bibfnamefont {C.~J.}\
  \bibnamefont {Wood}}, \bibinfo {author} {\bibfnamefont {J.-B.}\ \bibnamefont
  {Yau}}, \bibinfo {author} {\bibfnamefont {E.~J.}\ \bibnamefont {Zhang}},
  \bibinfo {author} {\bibfnamefont {O.~E.}\ \bibnamefont {Dial}}, \bibinfo
  {author} {\bibfnamefont {J.~M.}\ \bibnamefont {Chow}}, \ and\ \bibinfo
  {author} {\bibfnamefont {J.~M.}\ \bibnamefont {Gambetta}},\ }\bibfield
  {title} {\enquote {\bibinfo {title} {Demonstration of quantum volume 64 on a
  superconducting quantum computing system},}\ }\href {\doibase
  10.1088/2058-9565/abe519} {\bibfield  {journal} {\bibinfo  {journal} {Quantum
  Sci. Technol.}\ }\textbf {\bibinfo {volume} {6}},\ \bibinfo {pages} {025020}
  (\bibinfo {year} {2021})}\BibitemShut {NoStop}%
\bibitem [{\citenamefont {Cerezo}\ \emph {et~al.}(2021)\citenamefont {Cerezo},
  \citenamefont {Arrasmith}, \citenamefont {Babbush}, \citenamefont {Benjamin},
  \citenamefont {Endo}, \citenamefont {Fujii}, \citenamefont {McClean},
  \citenamefont {Mitarai}, \citenamefont {Yuan}, \citenamefont {Cincio},\ and\
  \citenamefont {Coles}}]{Cerezo_2021}%
  \BibitemOpen
  \bibfield  {author} {\bibinfo {author} {\bibfnamefont {M.}~\bibnamefont
  {Cerezo}}, \bibinfo {author} {\bibfnamefont {A.}~\bibnamefont {Arrasmith}},
  \bibinfo {author} {\bibfnamefont {R.}~\bibnamefont {Babbush}}, \bibinfo
  {author} {\bibfnamefont {S.~C.}\ \bibnamefont {Benjamin}}, \bibinfo {author}
  {\bibfnamefont {S.}~\bibnamefont {Endo}}, \bibinfo {author} {\bibfnamefont
  {K.}~\bibnamefont {Fujii}}, \bibinfo {author} {\bibfnamefont {J.~R.}\
  \bibnamefont {McClean}}, \bibinfo {author} {\bibfnamefont {K.}~\bibnamefont
  {Mitarai}}, \bibinfo {author} {\bibfnamefont {X.}~\bibnamefont {Yuan}},
  \bibinfo {author} {\bibfnamefont {L.}~\bibnamefont {Cincio}}, \ and\ \bibinfo
  {author} {\bibfnamefont {P.~J.}\ \bibnamefont {Coles}},\ }\bibfield  {title}
  {\enquote {\bibinfo {title} {Variational quantum algorithms},}\ }\href
  {\doibase 10.1038/s42254-021-00348-9} {\bibfield  {journal} {\bibinfo
  {journal} {Nature Rev. Phys.}\ }\textbf {\bibinfo {volume} {3}},\ \bibinfo
  {pages} {625--644} (\bibinfo {year} {2021})}\BibitemShut {NoStop}%
\bibitem [{\citenamefont {Peruzzo}\ \emph {et~al.}(2014)\citenamefont
  {Peruzzo}, \citenamefont {McClean}, \citenamefont {Shadbolt}, \citenamefont
  {Yung}, \citenamefont {Zhou}, \citenamefont {Love}, \citenamefont
  {Aspuru-Guzik},\ and\ \citenamefont {O'Brien}}]{Peruzzo}%
  \BibitemOpen
  \bibfield  {author} {\bibinfo {author} {\bibfnamefont {A.}~\bibnamefont
  {Peruzzo}}, \bibinfo {author} {\bibfnamefont {J.}~\bibnamefont {McClean}},
  \bibinfo {author} {\bibfnamefont {P.}~\bibnamefont {Shadbolt}}, \bibinfo
  {author} {\bibfnamefont {M.-H.}\ \bibnamefont {Yung}}, \bibinfo {author}
  {\bibfnamefont {X.-Q.}\ \bibnamefont {Zhou}}, \bibinfo {author}
  {\bibfnamefont {P.~J.}\ \bibnamefont {Love}}, \bibinfo {author}
  {\bibfnamefont {A.}~\bibnamefont {Aspuru-Guzik}}, \ and\ \bibinfo {author}
  {\bibfnamefont {J.~L.}\ \bibnamefont {O'Brien}},\ }\bibfield  {title}
  {\enquote {\bibinfo {title} {A variational eigenvalue solver on a photonic
  quantum processor},}\ }\href {\doibase 10.1038/ncomms5213} {\bibfield
  {journal} {\bibinfo  {journal} {Nature Comm.}\ }\textbf {\bibinfo {volume}
  {5}},\ \bibinfo {pages} {4213} (\bibinfo {year} {2014})}\BibitemShut
  {NoStop}%
\bibitem [{\citenamefont {Kandala}\ \emph {et~al.}(2017)\citenamefont
  {Kandala}, \citenamefont {Mezzcapo}, \citenamefont {Temme}, \citenamefont
  {Takita}, \citenamefont {Brink}, \citenamefont {Chow},\ and\ \citenamefont
  {Gambetta}}]{Kandala}%
  \BibitemOpen
  \bibfield  {author} {\bibinfo {author} {\bibfnamefont {A.}~\bibnamefont
  {Kandala}}, \bibinfo {author} {\bibfnamefont {A.}~\bibnamefont {Mezzcapo}},
  \bibinfo {author} {\bibfnamefont {K.}~\bibnamefont {Temme}}, \bibinfo
  {author} {\bibfnamefont {M.}~\bibnamefont {Takita}}, \bibinfo {author}
  {\bibfnamefont {M.}~\bibnamefont {Brink}}, \bibinfo {author} {\bibfnamefont
  {J.~W.}\ \bibnamefont {Chow}}, \ and\ \bibinfo {author} {\bibfnamefont
  {J.~M.}\ \bibnamefont {Gambetta}},\ }\bibfield  {title} {\enquote {\bibinfo
  {title} {Hardware-efficient variational quantum eigensolver for small
  molecules and quantum magnets},}\ }\href {\doibase 10.1038/nature23879}
  {\bibfield  {journal} {\bibinfo  {journal} {Nature}\ }\textbf {\bibinfo
  {volume} {549}},\ \bibinfo {pages} {242} (\bibinfo {year}
  {2017})}\BibitemShut {NoStop}%
\bibitem [{\citenamefont {McClean}\ \emph {et~al.}(2016)\citenamefont
  {McClean}, \citenamefont {Romero}, \citenamefont {Babbush},\ and\
  \citenamefont {Aspuru-Guzik}}]{McClean_2016}%
  \BibitemOpen
  \bibfield  {author} {\bibinfo {author} {\bibfnamefont {J.~R}\ \bibnamefont
  {McClean}}, \bibinfo {author} {\bibfnamefont {J.}~\bibnamefont {Romero}},
  \bibinfo {author} {\bibfnamefont {R.}~\bibnamefont {Babbush}}, \ and\
  \bibinfo {author} {\bibfnamefont {A.}~\bibnamefont {Aspuru-Guzik}},\
  }\bibfield  {title} {\enquote {\bibinfo {title} {The theory of variational
  hybrid quantum-classical algorithms},}\ }\href {\doibase
  10.1088/1367-2630/18/2/023023} {\bibfield  {journal} {\bibinfo  {journal}
  {New J. Phys.}\ }\textbf {\bibinfo {volume} {18}},\ \bibinfo {pages} {023023}
  (\bibinfo {year} {2016})}\BibitemShut {NoStop}%
\bibitem [{\citenamefont {Farhi}\ \emph {et~al.}(2014)\citenamefont {Farhi},
  \citenamefont {Goldstone},\ and\ \citenamefont {Gutmann}}]{QAOA}%
  \BibitemOpen
  \bibfield  {author} {\bibinfo {author} {\bibfnamefont {E.}~\bibnamefont
  {Farhi}}, \bibinfo {author} {\bibfnamefont {J.}~\bibnamefont {Goldstone}}, \
  and\ \bibinfo {author} {\bibfnamefont {S.}~\bibnamefont {Gutmann}},\
  }\bibfield  {title} {\enquote {\bibinfo {title} {A quantum approximate
  optimization algorithm},}\ }\href@noop {} {\bibfield  {journal} {\bibinfo
  {journal} {arxiv:1411.4028}\ } (\bibinfo {year} {2014})}\BibitemShut
  {NoStop}%
\bibitem [{\citenamefont {Zhou}\ \emph {et~al.}(2020)\citenamefont {Zhou},
  \citenamefont {Wang}, \citenamefont {Choi}, \citenamefont {Pichler},\ and\
  \citenamefont {Lukin}}]{Lukin}%
  \BibitemOpen
  \bibfield  {author} {\bibinfo {author} {\bibfnamefont {L.}~\bibnamefont
  {Zhou}}, \bibinfo {author} {\bibfnamefont {S.-T.}\ \bibnamefont {Wang}},
  \bibinfo {author} {\bibfnamefont {S.}~\bibnamefont {Choi}}, \bibinfo {author}
  {\bibfnamefont {H.}~\bibnamefont {Pichler}}, \ and\ \bibinfo {author}
  {\bibfnamefont {M.}~\bibnamefont {Lukin}},\ }\bibfield  {title} {\enquote
  {\bibinfo {title} {Quantum approximate optimization algorithm: Performance,
  mechanism, and implementation on near-term devices},}\ }\href {\doibase
  10.1103/PhysRevX.10.021067} {\bibfield  {journal} {\bibinfo  {journal} {Phys.
  Rev. X}\ }\textbf {\bibinfo {volume} {10}},\ \bibinfo {pages} {021067}
  (\bibinfo {year} {2020})}\BibitemShut {NoStop}%
\bibitem [{\citenamefont {Bharti}\ \emph {et~al.}(2022)\citenamefont {Bharti},
  \citenamefont {Cervera-Lierta}, \citenamefont {Kyaw}, \citenamefont {Haug},
  \citenamefont {Alperin-Lea}, \citenamefont {Anand}, \citenamefont {Degroote},
  \citenamefont {Heimonen}, \citenamefont {Kottmann}, \citenamefont {Menke}
  \emph {et~al.}}]{bharti_2021_noisy}%
  \BibitemOpen
  \bibfield  {author} {\bibinfo {author} {\bibfnamefont {K.}~\bibnamefont
  {Bharti}}, \bibinfo {author} {\bibfnamefont {A.}~\bibnamefont
  {Cervera-Lierta}}, \bibinfo {author} {\bibfnamefont {T.~H.}\ \bibnamefont
  {Kyaw}}, \bibinfo {author} {\bibfnamefont {T.}~\bibnamefont {Haug}}, \bibinfo
  {author} {\bibfnamefont {S.}~\bibnamefont {Alperin-Lea}}, \bibinfo {author}
  {\bibfnamefont {A.}~\bibnamefont {Anand}}, \bibinfo {author} {\bibfnamefont
  {M.}~\bibnamefont {Degroote}}, \bibinfo {author} {\bibfnamefont
  {H.}~\bibnamefont {Heimonen}}, \bibinfo {author} {\bibfnamefont {J.~S.}\
  \bibnamefont {Kottmann}}, \bibinfo {author} {\bibfnamefont {T.}~\bibnamefont
  {Menke}},  \emph {et~al.},\ }\bibfield  {title} {\enquote {\bibinfo {title}
  {{Noisy intermediate-scale quantum (NISQ) algorithms}},}\ }\href {\doibase
  10.48550/arXiv.2101.08448} {\bibfield  {journal} {\bibinfo  {journal} {Rev.
  Mod. Phys.}\ }\textbf {\bibinfo {volume} {94}},\ \bibinfo {pages} {015004}
  (\bibinfo {year} {2022})}\BibitemShut {NoStop}%
\bibitem [{\citenamefont {Tilly}\ \emph {et~al.}(2022)\citenamefont {Tilly},
  \citenamefont {Chen}, \citenamefont {Cao}, \citenamefont {Picozzi},
  \citenamefont {Setia}, \citenamefont {Li}, \citenamefont {Grant},
  \citenamefont {Wossnig}, \citenamefont {Rungger}, \citenamefont {Booth},\
  and\ \citenamefont {Tennyson}}]{VariationalReview}%
  \BibitemOpen
  \bibfield  {author} {\bibinfo {author} {\bibfnamefont {J.}~\bibnamefont
  {Tilly}}, \bibinfo {author} {\bibfnamefont {H.}~\bibnamefont {Chen}},
  \bibinfo {author} {\bibfnamefont {S.}~\bibnamefont {Cao}}, \bibinfo {author}
  {\bibfnamefont {D.}~\bibnamefont {Picozzi}}, \bibinfo {author} {\bibfnamefont
  {K.}~\bibnamefont {Setia}}, \bibinfo {author} {\bibfnamefont
  {Y.}~\bibnamefont {Li}}, \bibinfo {author} {\bibfnamefont {E.}~\bibnamefont
  {Grant}}, \bibinfo {author} {\bibfnamefont {L.}~\bibnamefont {Wossnig}},
  \bibinfo {author} {\bibfnamefont {I.}~\bibnamefont {Rungger}}, \bibinfo
  {author} {\bibfnamefont {G.~H.}\ \bibnamefont {Booth}}, \ and\ \bibinfo
  {author} {\bibfnamefont {J.}~\bibnamefont {Tennyson}},\ }\bibfield  {title}
  {\enquote {\bibinfo {title} {The variational quantum eigensolver: a review of
  methods and best practices},}\ }\href {\doibase
  10.1016/j.physrep.2022.08.003} {\bibfield  {journal} {\bibinfo  {journal}
  {Phys. Rep.}\ }\textbf {\bibinfo {volume} {986}},\ \bibinfo {pages} {1--128}
  (\bibinfo {year} {2022})}\BibitemShut {NoStop}%
\bibitem [{\citenamefont {Preskill}(2018)}]{preskill_quantum_2018}%
  \BibitemOpen
  \bibfield  {author} {\bibinfo {author} {\bibfnamefont {J.}~\bibnamefont
  {Preskill}},\ }\bibfield  {title} {\enquote {\bibinfo {title} {Quantum
  {computing} in the {NISQ} era and beyond},}\ }\href@noop {} {\  (\bibinfo
  {year} {2018})},\ \Eprint {http://arxiv.org/abs/1801.00862}
  {arXiv:1801.00862} \BibitemShut {NoStop}%
\bibitem [{\citenamefont {Schuld}\ \emph
  {et~al.}(2019{\natexlab{a}})\citenamefont {Schuld}, \citenamefont {Bergholm},
  \citenamefont {Gogolin}, \citenamefont {Izaac},\ and\ \citenamefont
  {Killoran}}]{PhysRevA.99.032331}%
  \BibitemOpen
  \bibfield  {author} {\bibinfo {author} {\bibfnamefont {M.}~\bibnamefont
  {Schuld}}, \bibinfo {author} {\bibfnamefont {V.}~\bibnamefont {Bergholm}},
  \bibinfo {author} {\bibfnamefont {C.}~\bibnamefont {Gogolin}}, \bibinfo
  {author} {\bibfnamefont {J.}~\bibnamefont {Izaac}}, \ and\ \bibinfo {author}
  {\bibfnamefont {N.}~\bibnamefont {Killoran}},\ }\bibfield  {title} {\enquote
  {\bibinfo {title} {Evaluating analytic gradients on quantum hardware},}\
  }\href {\doibase 10.1103/PhysRevA.99.032331} {\bibfield  {journal} {\bibinfo
  {journal} {Phys. Rev. A}\ }\textbf {\bibinfo {volume} {99}},\ \bibinfo
  {pages} {032331} (\bibinfo {year} {2019}{\natexlab{a}})}\BibitemShut
  {NoStop}%
\bibitem [{\citenamefont {Bergholm}\ \emph {et~al.}(2018)\citenamefont
  {Bergholm}, \citenamefont {Izaac}, \citenamefont {Schuld}, \citenamefont
  {Gogolin},\ and\ \citenamefont {Killoran}}]{bergholm2018pennylane}%
  \BibitemOpen
  \bibfield  {author} {\bibinfo {author} {\bibfnamefont {V.}~\bibnamefont
  {Bergholm}}, \bibinfo {author} {\bibfnamefont {J.}~\bibnamefont {Izaac}},
  \bibinfo {author} {\bibfnamefont {M.}~\bibnamefont {Schuld}}, \bibinfo
  {author} {\bibfnamefont {C.}~\bibnamefont {Gogolin}}, \ and\ \bibinfo
  {author} {\bibfnamefont {N.}~\bibnamefont {Killoran}},\ }\bibfield  {title}
  {\enquote {\bibinfo {title} {Pennylane: Automatic differentiation of hybrid
  quantum-classical computations},}\ }\href@noop {} {\  (\bibinfo {year}
  {2018})},\ \Eprint {http://arxiv.org/abs/1811.04968} {arXiv:1811.04968}
  \BibitemShut {NoStop}%
\bibitem [{pen()}]{pennylane}%
  \BibitemOpen
  \href@noop {} {}\bibinfo {note}
  {\url{https://pennylane.ai/qml/demos/tutorial\_vqe\_qng.html\#stokes2019}}\BibitemShut
  {NoStop}%
\bibitem [{\citenamefont {Sweke}\ \emph {et~al.}(2020)\citenamefont {Sweke},
  \citenamefont {Wilde}, \citenamefont {Meyer}, \citenamefont {Schuld},
  \citenamefont {F{\"a}hrmann}, \citenamefont {Meynard-Piganeau},\ and\
  \citenamefont {Eisert}}]{Gradients}%
  \BibitemOpen
  \bibfield  {author} {\bibinfo {author} {\bibfnamefont {R.}~\bibnamefont
  {Sweke}}, \bibinfo {author} {\bibfnamefont {F.}~\bibnamefont {Wilde}},
  \bibinfo {author} {\bibfnamefont {J.}~\bibnamefont {Meyer}}, \bibinfo
  {author} {\bibfnamefont {M.}~\bibnamefont {Schuld}}, \bibinfo {author}
  {\bibfnamefont {P.~K.}\ \bibnamefont {F{\"a}hrmann}}, \bibinfo {author}
  {\bibfnamefont {B.}~\bibnamefont {Meynard-Piganeau}}, \ and\ \bibinfo
  {author} {\bibfnamefont {J.}~\bibnamefont {Eisert}},\ }\bibfield  {title}
  {\enquote {\bibinfo {title} {Stochastic gradient descent for hybrid
  quantum-classical optimization},}\ }\href {\doibase
  10.22331/q-2020-08-31-314} {\bibfield  {journal} {\bibinfo  {journal}
  {Quantum}\ }\textbf {\bibinfo {volume} {4}},\ \bibinfo {pages} {314}
  (\bibinfo {year} {2020})}\BibitemShut {NoStop}%
\bibitem [{\citenamefont {Bray}\ and\ \citenamefont
  {Dean}(2007)}]{bray_statistics_2007}%
  \BibitemOpen
  \bibfield  {author} {\bibinfo {author} {\bibfnamefont {Alan~J.}\ \bibnamefont
  {Bray}}\ and\ \bibinfo {author} {\bibfnamefont {David~S.}\ \bibnamefont
  {Dean}},\ }\bibfield  {title} {\enquote {\bibinfo {title} {Statistics of
  {Critical} {Points} of {Gaussian} {Fields} on {Large}-{Dimensional}
  {Spaces}},}\ }\href {\doibase 10.1103/PhysRevLett.98.150201} {\bibfield
  {journal} {\bibinfo  {journal} {Physical Review Letters}\ }\textbf {\bibinfo
  {volume} {98}},\ \bibinfo {pages} {150201} (\bibinfo {year}
  {2007})}\BibitemShut {NoStop}%
\bibitem [{\citenamefont {Bittel}\ and\ \citenamefont
  {Kliesch}(2021)}]{Bittel}%
  \BibitemOpen
  \bibfield  {author} {\bibinfo {author} {\bibfnamefont {L.}~\bibnamefont
  {Bittel}}\ and\ \bibinfo {author} {\bibfnamefont {M.}~\bibnamefont
  {Kliesch}},\ }\bibfield  {title} {\enquote {\bibinfo {title} {{Training
  variational quantum algorithms is NP-hard}},}\ }\href {\doibase
  10.1103/PhysRevLett.127.120502} {\bibfield  {journal} {\bibinfo  {journal}
  {Phys. Rev. Lett.}\ }\textbf {\bibinfo {volume} {127}},\ \bibinfo {pages}
  {120502} (\bibinfo {year} {2021})}\BibitemShut {NoStop}%
\bibitem [{\citenamefont {Lee}\ \emph {et~al.}(2016)\citenamefont {Lee},
  \citenamefont {Simchowitz}, \citenamefont {Jordan},\ and\ \citenamefont
  {Recht}}]{lee2016gradient}%
  \BibitemOpen
  \bibfield  {author} {\bibinfo {author} {\bibfnamefont {Jason~D}\ \bibnamefont
  {Lee}}, \bibinfo {author} {\bibfnamefont {Max}\ \bibnamefont {Simchowitz}},
  \bibinfo {author} {\bibfnamefont {Michael~I}\ \bibnamefont {Jordan}}, \ and\
  \bibinfo {author} {\bibfnamefont {Benjamin}\ \bibnamefont {Recht}},\
  }\bibfield  {title} {\enquote {\bibinfo {title} {Gradient descent converges
  to minimizers},}\ }\href@noop {} {\  (\bibinfo {year} {2016})},\ \Eprint
  {http://arxiv.org/abs/1602.04915} {arXiv:1602.04915} \BibitemShut {NoStop}%
\bibitem [{\citenamefont {Du}\ \emph {et~al.}(2017{\natexlab{a}})\citenamefont
  {Du}, \citenamefont {Jin}, \citenamefont {Lee}, \citenamefont {Jordan},
  \citenamefont {Singh},\ and\ \citenamefont {Poczos}}]{du2017gradient}%
  \BibitemOpen
  \bibfield  {author} {\bibinfo {author} {\bibfnamefont {Simon~S}\ \bibnamefont
  {Du}}, \bibinfo {author} {\bibfnamefont {Chi}\ \bibnamefont {Jin}}, \bibinfo
  {author} {\bibfnamefont {Jason~D}\ \bibnamefont {Lee}}, \bibinfo {author}
  {\bibfnamefont {Michael~I.}\ \bibnamefont {Jordan}}, \bibinfo {author}
  {\bibfnamefont {Aarti}\ \bibnamefont {Singh}}, \ and\ \bibinfo {author}
  {\bibfnamefont {Barnabas}\ \bibnamefont {Poczos}},\ }\bibfield  {title}
  {\enquote {\bibinfo {title} {Gradient descent can take exponential time to
  escape saddle points},}\ }\href {\doibase 10.48550/arXiv.1706.02515}
  {\bibfield  {journal} {\bibinfo  {journal} {Adv. Neur. Inf. Proc. Sys.}\
  }\textbf {\bibinfo {volume} {30}} (\bibinfo {year} {2017}{\natexlab{a}}),\
  10.48550/arXiv.1706.02515}\BibitemShut {NoStop}%
\bibitem [{\citenamefont {Jin}\ \emph {et~al.}(2019)\citenamefont {Jin},
  \citenamefont {Netrapalli}, \citenamefont {Ge}, \citenamefont {Kakade},\ and\
  \citenamefont {Jordan}}]{JordanSaddlePoints}%
  \BibitemOpen
  \bibfield  {author} {\bibinfo {author} {\bibfnamefont {C.}~\bibnamefont
  {Jin}}, \bibinfo {author} {\bibfnamefont {P.}~\bibnamefont {Netrapalli}},
  \bibinfo {author} {\bibfnamefont {R.}~\bibnamefont {Ge}}, \bibinfo {author}
  {\bibfnamefont {S.~M.}\ \bibnamefont {Kakade}}, \ and\ \bibinfo {author}
  {\bibfnamefont {M.~I.}\ \bibnamefont {Jordan}},\ }\bibfield  {title}
  {\enquote {\bibinfo {title} {On nonconvex optimization for machine learning:
  Gradients, stochasticity, and saddle points},}\ }\href@noop {} {\bibfield
  {journal} {\bibinfo  {journal} {arXiv:1902.04811}\ } (\bibinfo {year}
  {2019})}\BibitemShut {NoStop}%
\bibitem [{\citenamefont {Jain}\ and\ \citenamefont {Kar}(2017)}]{jain2017non}%
  \BibitemOpen
  \bibfield  {author} {\bibinfo {author} {\bibfnamefont {P.}~\bibnamefont
  {Jain}}\ and\ \bibinfo {author} {\bibfnamefont {P.}~\bibnamefont {Kar}},\
  }\bibfield  {title} {\enquote {\bibinfo {title} {Non-convex optimization for
  machine learning},}\ }\href@noop {} {\  (\bibinfo {year} {2017})},\ \Eprint
  {http://arxiv.org/abs/1712.07897} {arXiv:1712.07897} \BibitemShut {NoStop}%
\bibitem [{\citenamefont {Duffield}\ \emph {et~al.}(2023)\citenamefont
  {Duffield}, \citenamefont {Benedetti},\ and\ \citenamefont
  {Rosenkranz}}]{duffield_bayesian_2022}%
  \BibitemOpen
  \bibfield  {author} {\bibinfo {author} {\bibfnamefont {S.}~\bibnamefont
  {Duffield}}, \bibinfo {author} {\bibfnamefont {M.}~\bibnamefont {Benedetti}},
  \ and\ \bibinfo {author} {\bibfnamefont {M.}~\bibnamefont {Rosenkranz}},\
  }\bibfield  {title} {\enquote {\bibinfo {title} {Bayesian {learning} of
  {parameterised} {quantum} {circuits}},}\ }\href {\doibase
  10.1088/2632-2153/acc8b7} {\bibfield  {journal} {\bibinfo  {journal} {Mach.
  Learn. Sci. Technol}\ }\textbf {\bibinfo {volume} {4}},\ \bibinfo {pages}
  {025007} (\bibinfo {year} {2023})}\BibitemShut {NoStop}%
\bibitem [{\citenamefont {Borras}\ \emph {et~al.}()\citenamefont {Borras},
  \citenamefont {Chang}, \citenamefont {Funcke}, \citenamefont {Grossi},
  \citenamefont {Hartung}, \citenamefont {Jansen}, \citenamefont {K{\"u}hn},
  \citenamefont {Rehm}, \citenamefont {T{\"u}ys{\"u}z},\ and\ \citenamefont
  {Vallecorsa}}]{Jansen}%
  \BibitemOpen
  \bibfield  {author} {\bibinfo {author} {\bibfnamefont {K.}~\bibnamefont
  {Borras}}, \bibinfo {author} {\bibfnamefont {S.~Y.}\ \bibnamefont {Chang}},
  \bibinfo {author} {\bibfnamefont {L.}~\bibnamefont {Funcke}}, \bibinfo
  {author} {\bibfnamefont {M.}~\bibnamefont {Grossi}}, \bibinfo {author}
  {\bibfnamefont {T.}~\bibnamefont {Hartung}}, \bibinfo {author} {\bibfnamefont
  {K.}~\bibnamefont {Jansen}}, \bibinfo {author} {\bibfnamefont
  {D.~Kruecker~S.}\ \bibnamefont {K{\"u}hn}}, \bibinfo {author} {\bibfnamefont
  {F.}~\bibnamefont {Rehm}}, \bibinfo {author} {\bibfnamefont {C.}~\bibnamefont
  {T{\"u}ys{\"u}z}}, \ and\ \bibinfo {author} {\bibfnamefont {S.}~\bibnamefont
  {Vallecorsa}},\ }\bibfield  {title} {\enquote {\bibinfo {title} {Impact of
  quantum noise on the training of quantum generative adversarial networks},}\
  }\href {\doibase 10.1088/1742-6596/2438/1/012093} {\bibfield  {journal}
  {\bibinfo  {journal} {J. Phys. Conf. Ser.}\ }\textbf {\bibinfo {volume}
  {2438}},\ \bibinfo {pages} {012093}}\BibitemShut {NoStop}%
\bibitem [{\citenamefont {Duffield}\ \emph {et~al.}(2022)\citenamefont
  {Duffield}, \citenamefont {Benedetti},\ and\ \citenamefont
  {Rosenkranz}}]{Rosenkranz}%
  \BibitemOpen
  \bibfield  {author} {\bibinfo {author} {\bibfnamefont {S.}~\bibnamefont
  {Duffield}}, \bibinfo {author} {\bibfnamefont {M.}~\bibnamefont {Benedetti}},
  \ and\ \bibinfo {author} {\bibfnamefont {M.}~\bibnamefont {Rosenkranz}},\
  }\bibfield  {title} {\enquote {\bibinfo {title} {Bayesian learning of
  parameterised quantum circuits},}\ }\href@noop {} {\  (\bibinfo {year}
  {2022})},\ \Eprint {http://arxiv.org/abs/2206.07559} {arXiv:2206.07559}
  \BibitemShut {NoStop}%
\bibitem [{\citenamefont {Oliv}\ \emph {et~al.}(2022)\citenamefont {Oliv},
  \citenamefont {Matic}, \citenamefont {Messerer},\ and\ \citenamefont
  {Lorenz}}]{Lorenz}%
  \BibitemOpen
  \bibfield  {author} {\bibinfo {author} {\bibfnamefont {M.}~\bibnamefont
  {Oliv}}, \bibinfo {author} {\bibfnamefont {A.}~\bibnamefont {Matic}},
  \bibinfo {author} {\bibfnamefont {T.}~\bibnamefont {Messerer}}, \ and\
  \bibinfo {author} {\bibfnamefont {J.~Miriam}\ \bibnamefont {Lorenz}},\
  }\bibfield  {title} {\enquote {\bibinfo {title} {Evaluating the impact of
  noise on the performance of the variational quantum eigensolver},}\
  }\href@noop {} {\  (\bibinfo {year} {2022})},\ \Eprint
  {http://arxiv.org/abs/2209.12803} {arXiv:2209.12803} \BibitemShut {NoStop}%
\bibitem [{\citenamefont {Patti}\ \emph {et~al.}(2020)\citenamefont {Patti},
  \citenamefont {Najafi}, \citenamefont {Gao},\ and\ \citenamefont
  {Yelin}}]{Yelin2}%
  \BibitemOpen
  \bibfield  {author} {\bibinfo {author} {\bibfnamefont {T.~L.}\ \bibnamefont
  {Patti}}, \bibinfo {author} {\bibfnamefont {K.}~\bibnamefont {Najafi}},
  \bibinfo {author} {\bibfnamefont {X.}~\bibnamefont {Gao}}, \ and\ \bibinfo
  {author} {\bibfnamefont {S.~F.}\ \bibnamefont {Yelin}},\ }\bibfield  {title}
  {\enquote {\bibinfo {title} {Entanglement devised barren plateau
  mitigation},}\ }\href@noop {} {\bibfield  {journal} {\bibinfo  {journal}
  {arXiv:2012.12658}\ } (\bibinfo {year} {2020})}\BibitemShut {NoStop}%
\bibitem [{\citenamefont {Gu}\ \emph {et~al.}(2021)\citenamefont {Gu},
  \citenamefont {Lowe}, \citenamefont {Dub}, \citenamefont {Coles},\ and\
  \citenamefont {Arrasmith}}]{gu_adaptive_2021}%
  \BibitemOpen
  \bibfield  {author} {\bibinfo {author} {\bibfnamefont {Andi}\ \bibnamefont
  {Gu}}, \bibinfo {author} {\bibfnamefont {Angus}\ \bibnamefont {Lowe}},
  \bibinfo {author} {\bibfnamefont {Pavel~A.}\ \bibnamefont {Dub}}, \bibinfo
  {author} {\bibfnamefont {Patrick~J.}\ \bibnamefont {Coles}}, \ and\ \bibinfo
  {author} {\bibfnamefont {Andrew}\ \bibnamefont {Arrasmith}},\ }\bibfield
  {title} {\enquote {\bibinfo {title} {Adaptive shot allocation for fast
  convergence in variational quantum algorithms},}\ }\href@noop {} {\
  (\bibinfo {year} {2021})},\ \Eprint {http://arxiv.org/abs/2108.10434}
  {arXiv:2108.10434} \BibitemShut {NoStop}%
\bibitem [{\citenamefont {Gentini}\ \emph {et~al.}(2020)\citenamefont
  {Gentini}, \citenamefont {Cuccoli}, \citenamefont {Pirandola}, \citenamefont
  {Verrucchi},\ and\ \citenamefont {Banchi}}]{PhysRevA.102.052414}%
  \BibitemOpen
  \bibfield  {author} {\bibinfo {author} {\bibfnamefont {L.}~\bibnamefont
  {Gentini}}, \bibinfo {author} {\bibfnamefont {A.}~\bibnamefont {Cuccoli}},
  \bibinfo {author} {\bibfnamefont {S.}~\bibnamefont {Pirandola}}, \bibinfo
  {author} {\bibfnamefont {P.}~\bibnamefont {Verrucchi}}, \ and\ \bibinfo
  {author} {\bibfnamefont {L.}~\bibnamefont {Banchi}},\ }\bibfield  {title}
  {\enquote {\bibinfo {title} {Noise-resilient variational hybrid
  quantum-classical optimization},}\ }\href {\doibase
  10.1103/PhysRevA.102.052414} {\bibfield  {journal} {\bibinfo  {journal}
  {Phys. Rev. A}\ }\textbf {\bibinfo {volume} {102}},\ \bibinfo {pages}
  {052414} (\bibinfo {year} {2020})}\BibitemShut {NoStop}%
\bibitem [{\citenamefont {De~Palma}\ \emph {et~al.}(2022)\citenamefont
  {De~Palma}, \citenamefont {Marvian}, \citenamefont {Rouzé},\ and\
  \citenamefont {França}}]{DePalmaVQAs}%
  \BibitemOpen
  \bibfield  {author} {\bibinfo {author} {\bibfnamefont {G.}~\bibnamefont
  {De~Palma}}, \bibinfo {author} {\bibfnamefont {M.}~\bibnamefont {Marvian}},
  \bibinfo {author} {\bibfnamefont {C.}~\bibnamefont {Rouzé}}, \ and\ \bibinfo
  {author} {\bibfnamefont {D.~Stilck}\ \bibnamefont {França}},\ }\bibfield
  {title} {\enquote {\bibinfo {title} {Limitations of variational quantum
  algorithms: a quantum optimal transport approach},}\ }\href@noop {} {\
  (\bibinfo {year} {2022})},\ \Eprint {http://arxiv.org/abs/2204.03455}
  {2204.03455} \BibitemShut {NoStop}%
\bibitem [{\citenamefont {Fran{\c{c}}a}\ and\ \citenamefont
  {Garc{\'{\i}}a-Patr{\'{o}}n}(2021)}]{Stilck_Fran_a_2021}%
  \BibitemOpen
  \bibfield  {author} {\bibinfo {author} {\bibfnamefont {D.~Stilck}\
  \bibnamefont {Fran{\c{c}}a}}\ and\ \bibinfo {author} {\bibfnamefont
  {R.}~\bibnamefont {Garc{\'{\i}}a-Patr{\'{o}}n}},\ }\bibfield  {title}
  {\enquote {\bibinfo {title} {Limitations of optimization algorithms on noisy
  quantum devices},}\ }\href {\doibase 10.1038/s41567-021-01356-3} {\bibfield
  {journal} {\bibinfo  {journal} {Nature Phys.}\ }\textbf {\bibinfo {volume}
  {17}},\ \bibinfo {pages} {1221--1227} (\bibinfo {year} {2021})}\BibitemShut
  {NoStop}%
\bibitem [{\citenamefont {Wang}\ \emph {et~al.}(2021)\citenamefont {Wang},
  \citenamefont {Fontana}, \citenamefont {Cerezo}, \citenamefont {Sharma},
  \citenamefont {Sone}, \citenamefont {Cincio},\ and\ \citenamefont
  {Coles}}]{NoisyBP_Wang_2021}%
  \BibitemOpen
  \bibfield  {author} {\bibinfo {author} {\bibfnamefont {S.}~\bibnamefont
  {Wang}}, \bibinfo {author} {\bibfnamefont {E.}~\bibnamefont {Fontana}},
  \bibinfo {author} {\bibfnamefont {M.}~\bibnamefont {Cerezo}}, \bibinfo
  {author} {\bibfnamefont {K.}~\bibnamefont {Sharma}}, \bibinfo {author}
  {\bibfnamefont {A.}~\bibnamefont {Sone}}, \bibinfo {author} {\bibfnamefont
  {L.}~\bibnamefont {Cincio}}, \ and\ \bibinfo {author} {\bibfnamefont {P.~J.}\
  \bibnamefont {Coles}},\ }\bibfield  {title} {\enquote {\bibinfo {title}
  {Noise-induced barren plateaus in variational quantum algorithms},}\ }\href
  {\doibase 10.1038/s41467-021-27045-6} {\bibfield  {journal} {\bibinfo
  {journal} {Nature Comm.}\ }\textbf {\bibinfo {volume} {12}},\ \bibinfo
  {pages} {6961} (\bibinfo {year} {2021})}\BibitemShut {NoStop}%
\bibitem [{\citenamefont {McClean}\ \emph {et~al.}(2018)\citenamefont
  {McClean}, \citenamefont {Boixo}, \citenamefont {Smelyanskiy}, \citenamefont
  {Babbush},\ and\ \citenamefont {Neven}}]{McClean_2018}%
  \BibitemOpen
  \bibfield  {author} {\bibinfo {author} {\bibfnamefont {J.~R.}\ \bibnamefont
  {McClean}}, \bibinfo {author} {\bibfnamefont {S.}~\bibnamefont {Boixo}},
  \bibinfo {author} {\bibfnamefont {V.~N.}\ \bibnamefont {Smelyanskiy}},
  \bibinfo {author} {\bibfnamefont {R.}~\bibnamefont {Babbush}}, \ and\
  \bibinfo {author} {\bibfnamefont {H.}~\bibnamefont {Neven}},\ }\bibfield
  {title} {\enquote {\bibinfo {title} {Barren plateaus in quantum neural
  network training landscapes},}\ }\href {\doibase 10.1038/s41467-018-07090-4}
  {\bibfield  {journal} {\bibinfo  {journal} {Nature Comm.}\ }\textbf {\bibinfo
  {volume} {9}},\ \bibinfo {pages} {4812} (\bibinfo {year} {2018})}\BibitemShut
  {NoStop}%
\bibitem [{\citenamefont {Du}\ \emph {et~al.}(2017{\natexlab{b}})\citenamefont
  {Du}, \citenamefont {Jin}, \citenamefont {Lee}, \citenamefont {Jordan},
  \citenamefont {Poczos},\ and\ \citenamefont {Singh}}]{ExpTimeSaddle}%
  \BibitemOpen
  \bibfield  {author} {\bibinfo {author} {\bibfnamefont {S.~S.}\ \bibnamefont
  {Du}}, \bibinfo {author} {\bibfnamefont {C.}~\bibnamefont {Jin}}, \bibinfo
  {author} {\bibfnamefont {J.~D.}\ \bibnamefont {Lee}}, \bibinfo {author}
  {\bibfnamefont {M.~I.}\ \bibnamefont {Jordan}}, \bibinfo {author}
  {\bibfnamefont {B.}~\bibnamefont {Poczos}}, \ and\ \bibinfo {author}
  {\bibfnamefont {A.}~\bibnamefont {Singh}},\ }\bibfield  {title} {\enquote
  {\bibinfo {title} {Gradient descent can take exponential time to escape
  saddle points},}\ }\href@noop {} {\bibfield  {journal} {\bibinfo  {journal}
  {arXiv:1705.10412}\ } (\bibinfo {year} {2017}{\natexlab{b}})},\ \Eprint
  {http://arxiv.org/abs/1705.10412} {arXiv:1705.10412} \BibitemShut {NoStop}%
\bibitem [{\citenamefont {Lee}\ \emph {et~al.}(2017)\citenamefont {Lee},
  \citenamefont {Panageas}, \citenamefont {Piliouras}, \citenamefont
  {Simchowitz}, \citenamefont {Jordan},\ and\ \citenamefont {Recht}}]{Lee2017}%
  \BibitemOpen
  \bibfield  {author} {\bibinfo {author} {\bibfnamefont {J.~D.}\ \bibnamefont
  {Lee}}, \bibinfo {author} {\bibfnamefont {I.}~\bibnamefont {Panageas}},
  \bibinfo {author} {\bibfnamefont {G.}~\bibnamefont {Piliouras}}, \bibinfo
  {author} {\bibfnamefont {M.}~\bibnamefont {Simchowitz}}, \bibinfo {author}
  {\bibfnamefont {M.~I.}\ \bibnamefont {Jordan}}, \ and\ \bibinfo {author}
  {\bibfnamefont {B.}~\bibnamefont {Recht}},\ }\bibfield  {title} {\enquote
  {\bibinfo {title} {First-order methods almost always avoid saddle points},}\
  }\href@noop {} {\  (\bibinfo {year} {2017})},\ \Eprint
  {http://arxiv.org/abs/1710.07406} {arXiv:1710.07406} \BibitemShut {NoStop}%
\bibitem [{\citenamefont {Schuld}\ \emph
  {et~al.}(2019{\natexlab{b}})\citenamefont {Schuld}, \citenamefont {Bergholm},
  \citenamefont {Gogolin}, \citenamefont {Izaac},\ and\ \citenamefont
  {Killoran}}]{GradientsQcomp}%
  \BibitemOpen
  \bibfield  {author} {\bibinfo {author} {\bibfnamefont {M.}~\bibnamefont
  {Schuld}}, \bibinfo {author} {\bibfnamefont {V.}~\bibnamefont {Bergholm}},
  \bibinfo {author} {\bibfnamefont {C.}~\bibnamefont {Gogolin}}, \bibinfo
  {author} {\bibfnamefont {J.}~\bibnamefont {Izaac}}, \ and\ \bibinfo {author}
  {\bibfnamefont {N.}~\bibnamefont {Killoran}},\ }\bibfield  {title} {\enquote
  {\bibinfo {title} {Evaluating analytic gradients on quantum hardware},}\
  }\href {\doibase 10.1103/PhysRevA.99.032331} {\bibfield  {journal} {\bibinfo
  {journal} {Phys. Rev. A}\ }\textbf {\bibinfo {volume} {99}},\ \bibinfo
  {pages} {032331} (\bibinfo {year} {2019}{\natexlab{b}})}\BibitemShut
  {NoStop}%
\bibitem [{\citenamefont {Mele}\ \emph {et~al.}(2024)\citenamefont {Mele},
  \citenamefont {Angrisani}, \citenamefont {Ghosh}, \citenamefont {Khatri},
  \citenamefont {Eisert}, \citenamefont {França},\ and\ \citenamefont
  {Quek}}]{Nonunital}%
  \BibitemOpen
  \bibfield  {author} {\bibinfo {author} {\bibfnamefont {A.~A.}\ \bibnamefont
  {Mele}}, \bibinfo {author} {\bibfnamefont {A.}~\bibnamefont {Angrisani}},
  \bibinfo {author} {\bibfnamefont {S.}~\bibnamefont {Ghosh}}, \bibinfo
  {author} {\bibfnamefont {S.}~\bibnamefont {Khatri}}, \bibinfo {author}
  {\bibfnamefont {J.}~\bibnamefont {Eisert}}, \bibinfo {author} {\bibfnamefont
  {D.~Stilck}\ \bibnamefont {França}}, \ and\ \bibinfo {author} {\bibfnamefont
  {Y.}~\bibnamefont {Quek}},\ }\bibfield  {title} {\enquote {\bibinfo {title}
  {Noise-induced shallow circuits and absence of barren plateaus},}\
  }\href@noop {} {\  (\bibinfo {year} {2024})},\ \Eprint
  {http://arxiv.org/abs/2403.13927} {2403.13927} \BibitemShut {NoStop}%
\bibitem [{\citenamefont {Benzi}\ \emph {et~al.}(1981)\citenamefont {Benzi},
  \citenamefont {Sutera},\ and\ \citenamefont
  {Vulpiani}}]{StochasticResonance}%
  \BibitemOpen
  \bibfield  {author} {\bibinfo {author} {\bibfnamefont {R.}~\bibnamefont
  {Benzi}}, \bibinfo {author} {\bibfnamefont {A.}~\bibnamefont {Sutera}}, \
  and\ \bibinfo {author} {\bibfnamefont {A.}~\bibnamefont {Vulpiani}},\
  }\bibfield  {title} {\enquote {\bibinfo {title} {The mechanism of stochastic
  resonance},}\ }\href {\doibase 10.1088/0305-4470/14/11/006} {\bibfield
  {journal} {\bibinfo  {journal} {J. Phys. A}\ }\textbf {\bibinfo {volume}
  {14}},\ \bibinfo {pages} {L453–L457} (\bibinfo {year} {1981})}\BibitemShut
  {NoStop}%
\bibitem [{\citenamefont {Lykov}\ \emph {et~al.}(2023)\citenamefont {Lykov},
  \citenamefont {Wurtz}, \citenamefont {Poole}, \citenamefont {Saffman},
  \citenamefont {Noel},\ and\ \citenamefont {Alexeev}}]{Alexeev}%
  \BibitemOpen
  \bibfield  {author} {\bibinfo {author} {\bibfnamefont {D.}~\bibnamefont
  {Lykov}}, \bibinfo {author} {\bibfnamefont {J.}~\bibnamefont {Wurtz}},
  \bibinfo {author} {\bibfnamefont {C.}~\bibnamefont {Poole}}, \bibinfo
  {author} {\bibfnamefont {M.}~\bibnamefont {Saffman}}, \bibinfo {author}
  {\bibfnamefont {T.}~\bibnamefont {Noel}}, \ and\ \bibinfo {author}
  {\bibfnamefont {Y.}~\bibnamefont {Alexeev}},\ }\bibfield  {title} {\enquote
  {\bibinfo {title} {Sampling frequency thresholds for the quantum advantage of
  the quantum approximate optimization algorithm},}\ }\href {\doibase
  10.1038/s41534-023-00718-4} {\bibfield  {journal} {\bibinfo  {journal} {npj
  Quant. Inf.}\ }\textbf {\bibinfo {volume} {9}},\ \bibinfo {pages} {73}
  (\bibinfo {year} {2023})}\BibitemShut {NoStop}%
\bibitem [{\citenamefont {Fontana}\ \emph {et~al.}(2021)\citenamefont
  {Fontana}, \citenamefont {Fitzpatrick}, \citenamefont {Ramo}, \citenamefont
  {Duncan},\ and\ \citenamefont {Rungger}}]{PhysRevA.104.022403}%
  \BibitemOpen
  \bibfield  {author} {\bibinfo {author} {\bibfnamefont {Enrico}\ \bibnamefont
  {Fontana}}, \bibinfo {author} {\bibfnamefont {Nathan}\ \bibnamefont
  {Fitzpatrick}}, \bibinfo {author} {\bibfnamefont {David~Munoz}\ \bibnamefont
  {Ramo}}, \bibinfo {author} {\bibfnamefont {Ross}\ \bibnamefont {Duncan}}, \
  and\ \bibinfo {author} {\bibfnamefont {Ivan}\ \bibnamefont {Rungger}},\
  }\bibfield  {title} {\enquote {\bibinfo {title} {Evaluating the noise
  resilience of variational quantum algorithms},}\ }\href {\doibase
  10.1103/PhysRevA.104.022403} {\bibfield  {journal} {\bibinfo  {journal}
  {Phys. Rev. A}\ }\textbf {\bibinfo {volume} {104}},\ \bibinfo {pages}
  {022403} (\bibinfo {year} {2021})}\BibitemShut {NoStop}%
\bibitem [{\citenamefont {Haug}\ \emph {et~al.}(2021)\citenamefont {Haug},
  \citenamefont {Bharti},\ and\ \citenamefont {Kim}}]{PRXQuantum.2.040309}%
  \BibitemOpen
  \bibfield  {author} {\bibinfo {author} {\bibfnamefont {Tobias}\ \bibnamefont
  {Haug}}, \bibinfo {author} {\bibfnamefont {Kishor}\ \bibnamefont {Bharti}}, \
  and\ \bibinfo {author} {\bibfnamefont {M.S.}\ \bibnamefont {Kim}},\
  }\bibfield  {title} {\enquote {\bibinfo {title} {Capacity and quantum
  geometry of parametrized quantum circuits},}\ }\href {\doibase
  10.1103/PRXQuantum.2.040309} {\bibfield  {journal} {\bibinfo  {journal} {PRX
  Quantum}\ }\textbf {\bibinfo {volume} {2}},\ \bibinfo {pages} {040309}
  (\bibinfo {year} {2021})}\BibitemShut {NoStop}%
\bibitem [{\citenamefont {Patel}\ \emph {et~al.}(2021)\citenamefont {Patel},
  \citenamefont {Coles},\ and\ \citenamefont {Wilde}}]{WildeSDP}%
  \BibitemOpen
  \bibfield  {author} {\bibinfo {author} {\bibfnamefont {D.}~\bibnamefont
  {Patel}}, \bibinfo {author} {\bibfnamefont {P.~J.}\ \bibnamefont {Coles}}, \
  and\ \bibinfo {author} {\bibfnamefont {M.~M.}\ \bibnamefont {Wilde}},\
  }\bibfield  {title} {\enquote {\bibinfo {title} {Variational quantum
  algorithms for semidefinite programming},}\ }\href@noop {} {\  (\bibinfo
  {year} {2021})},\ \Eprint {http://arxiv.org/abs/2112.08859}
  {arXiv:2112.08859} \BibitemShut {NoStop}%
\bibitem [{\citenamefont {P{\'o}lya}(1921)}]{polya1921aufgabe}%
  \BibitemOpen
  \bibfield  {author} {\bibinfo {author} {\bibfnamefont {G.}~\bibnamefont
  {P{\'o}lya}},\ }\bibfield  {title} {\enquote {\bibinfo {title} {{{\"U}ber
  eine Aufgabe der Wahrscheinlichkeitsrechnung betreffend die Irrfahrt im
  Stra{\ss}ennetz}},}\ }\href@noop {} {\bibfield  {journal} {\bibinfo
  {journal} {Math. Ann.}\ }\textbf {\bibinfo {volume} {84}},\ \bibinfo {pages}
  {149--160} (\bibinfo {year} {1921})}\BibitemShut {NoStop}%
\bibitem [{\citenamefont {Montroll}(1956)}]{montroll1956random}%
  \BibitemOpen
  \bibfield  {author} {\bibinfo {author} {\bibfnamefont {E.~W.}\ \bibnamefont
  {Montroll}},\ }\bibfield  {title} {\enquote {\bibinfo {title} {Random walks
  in multidimensional spaces, especially on periodic lattices},}\ }\href@noop
  {} {\bibfield  {journal} {\bibinfo  {journal} {J. Soc. Ind. Appl. Math.}\
  }\textbf {\bibinfo {volume} {4}},\ \bibinfo {pages} {241--260} (\bibinfo
  {year} {1956})}\BibitemShut {NoStop}%
\bibitem [{\citenamefont {Liu}\ \emph {et~al.}(2022{\natexlab{a}})\citenamefont
  {Liu}, \citenamefont {Tacchino}, \citenamefont {Glick}, \citenamefont
  {Jiang},\ and\ \citenamefont {Mezzacapo}}]{Liu:2021wqr}%
  \BibitemOpen
  \bibfield  {author} {\bibinfo {author} {\bibfnamefont {J.}~\bibnamefont
  {Liu}}, \bibinfo {author} {\bibfnamefont {F.}~\bibnamefont {Tacchino}},
  \bibinfo {author} {\bibfnamefont {J.~R.}\ \bibnamefont {Glick}}, \bibinfo
  {author} {\bibfnamefont {L.}~\bibnamefont {Jiang}}, \ and\ \bibinfo {author}
  {\bibfnamefont {A.}~\bibnamefont {Mezzacapo}},\ }\bibfield  {title} {\enquote
  {\bibinfo {title} {{Representation learning via quantum neural tangent
  kernels}},}\ }\href {\doibase 10.1103/PRXQuantum.3.030323} {\bibfield
  {journal} {\bibinfo  {journal} {PRX Quantum}\ }\textbf {\bibinfo {volume}
  {3}},\ \bibinfo {pages} {030323} (\bibinfo {year}
  {2022}{\natexlab{a}})}\BibitemShut {NoStop}%
\bibitem [{\citenamefont {Liu}\ \emph {et~al.}(2023)\citenamefont {Liu},
  \citenamefont {Najafi}, \citenamefont {Sharma}, \citenamefont {Tacchino},
  \citenamefont {Jiang},\ and\ \citenamefont {Mezzacapo}}]{Liu:2022eqa}%
  \BibitemOpen
  \bibfield  {author} {\bibinfo {author} {\bibfnamefont {J.}~\bibnamefont
  {Liu}}, \bibinfo {author} {\bibfnamefont {K.}~\bibnamefont {Najafi}},
  \bibinfo {author} {\bibfnamefont {K.}~\bibnamefont {Sharma}}, \bibinfo
  {author} {\bibfnamefont {F.}~\bibnamefont {Tacchino}}, \bibinfo {author}
  {\bibfnamefont {L.}~\bibnamefont {Jiang}}, \ and\ \bibinfo {author}
  {\bibfnamefont {A.}~\bibnamefont {Mezzacapo}},\ }\bibfield  {title} {\enquote
  {\bibinfo {title} {{An analytic theory for the dynamics of wide quantum
  neural networks}},}\ }\href {\doibase 10.1103/PhysRevLett.130.150601}
  {\bibfield  {journal} {\bibinfo  {journal} {Phys. Rev. Lett.}\ }\textbf
  {\bibinfo {volume} {130}},\ \bibinfo {pages} {150601} (\bibinfo {year}
  {2023})}\BibitemShut {NoStop}%
\bibitem [{\citenamefont {Liu}\ \emph {et~al.}(2022{\natexlab{b}})\citenamefont
  {Liu}, \citenamefont {Lin},\ and\ \citenamefont {Jiang}}]{Liu:2022rhw}%
  \BibitemOpen
  \bibfield  {author} {\bibinfo {author} {\bibfnamefont {J.}~\bibnamefont
  {Liu}}, \bibinfo {author} {\bibfnamefont {Z.}~\bibnamefont {Lin}}, \ and\
  \bibinfo {author} {\bibfnamefont {L.}~\bibnamefont {Jiang}},\ }\bibfield
  {title} {\enquote {\bibinfo {title} {{Laziness, barren plateau, and noise in
  machine learning}},}\ }\href@noop {} {\bibfield  {journal} {\bibinfo
  {journal} {arXiv:2206.09313}\ } (\bibinfo {year}
  {2022}{\natexlab{b}})}\BibitemShut {NoStop}%
\end{thebibliography}%
\newpage

\appendix


\section*{Appendix}

\subsection{Strong smoothness and Lipschitz-Hessian property}
\label{sec:smoothness}

\commentFre{In this section, we provide a proof of Theorem~14 of the main text.
As stated in the main text, we}
focus our analysis on ansatz circuits of the form
\begin{equation}
    \label{eq:circuit_structure}
    U({\theta}):=\prod_{\ell=1}^{p} W_{\ell} \exp \left(i \theta_{\ell} X_{\ell}\right)
\end{equation}
where $W_{\ell}$ and $X_{\ell}$ are respectively fixed unitaries and Hermitian operators.
{As a reminder, we want to show that the loss function
\begin{equation}
    \label{eq:loss_apdx}
    \mathcal{L}({\theta})=\left\langle 0\left| {{U}^{\dagger}}({\theta})OU({\theta} ) \right|0 \right\rangle
\end{equation}
is $\beta$-strongly smooth and has a $\rho$-Lipschitz Hessian, with
\begin{align}
    \beta & \le 2^2p\abs{O}_\infty\max_{j=1,\dots,p}\abs{X_{j}}_\infty^2, \\
    \rho & \le 2^3p^\frac{3}{2} \abs{O}_\infty\max_{j=1,\dots,p}\abs{X_{j}}_\infty^3.
\end{align}}

{To begin, we state three important facts about Lipschitz constants of multi-variate functions.
\begin{lemma}
    If $\mathcal{L}: \mathbb{R}^p \to \mathbb{R}$ is differentiable with bounded partial derivatives, then
    \begin{equation}
        L=\sqrt{p}\max_{j}\left(\sup_{{\theta}}\left|\frac{\partial \mathcal{L}({\theta})}{\partial \theta_j}\right|\right)
    \end{equation}
    is the Lipschitz constant for $\mathcal{L}$.
    \label{LemmaLipRptoR}
\end{lemma}}
{The proof is given in Ref.~\cite{WildeSDP} (Lemma 7). }

\begin{lemma}
    \label{lem:combined-lipschitz-constant}
    If $\mathcal{L}: \mathbb{R}^p \to \mathbb{R}^M$ is a function with all its $M$ components Lipschitz functions with Lipschitz constant $L_i$, then $\mathcal{L}$ has Lipschitz constant $L=\sqrt{\sum_{i=1}^M L_i^2}$.
    \label{LemmaLipRptoRm}
\end{lemma}
{The proof is given in Ref.~\cite{WildeSDP} (Lemma 8).}
{Equipped with these facts we can proceed to derive an upper bound for the Lipschitz constant of functions from $\mathbb{R}^p$ to $\mathbb{R}^M$.}

\begin{lemma}
If $g: \mathbb{R}^p \to \mathbb{R}^M$ is a differentiable function with bounded gradient, then its Lipschitz constant {$L$ satisfies}
\begin{equation}
    L\le \sqrt{p M} \max_{i,j}\left(\sup_{{\theta}}\left| \frac{\partial g_i({\theta}) }{\partial \theta_j } \right|\right),
\end{equation}
where $g_i({\theta})$ the $i$-th component of $\theta\mapsto g({\theta})$.
\label{LemmaRPtoRMgeneral}
\end{lemma}

\begin{proof}
Using Lemma $\ref{LemmaLipRptoR}$ and $\ref{LemmaLipRptoRm}$, we have
\begin{align}
    L = \left({\sum_{i=1}^M L_i^2}\right)^{1/2} 
    &\le \sqrt{M} \max_i(L_i) \nonumber \\
    &= \sqrt{p M} \max_{i,j}\left(\sup_{{\theta}}\left| \frac{\partial g_i({\theta}) }{\partial \theta_j } \right|\right),
\end{align}
{where $L_i$ is the Lipschitz constant of the $i$-th component of $\mathcal{L}$ as defined in Lemma \ref{lem:combined-lipschitz-constant}.}
\end{proof}
{Next, we focus on loss functions of the type in equation~(\ref{eq:loss_apdx}).}
\begin{lemma}
    \label{SupDeriv}
    The loss function {as defined in Eq.~(\ref{eq:loss_apdx})} (with ${\theta} \in \mathbb{R}^p$) \commentFre{satisfies}
    \begin{align}
        \max_{i_1,i_2\dots i_k}\left(\sup_{{\theta}}\left|\frac{\partial^k \mathcal{L}({\theta})}{\partial \theta_{i_k} \cdots \partial \theta_{i_2}\partial \theta_{i_1}}\right|\right)\le 2^k \abs{O}_\infty\max_{j=1,\dots,p}\abs{X_{j}}_\infty^k
    \end{align}
    where {$U(\theta)$ is given by Eq.~(\ref{eq:circuit_structure})}.
\end{lemma}
\begin{proof}
We introduce the standard \emph{multi-index} notation. \newje{For this,}  we have
\begin{align}
   \partial^\alpha \mathcal{L}({\theta}) := \frac{\partial^k \mathcal{L}({\theta})}{\partial \theta_{i_k} \cdots \partial \theta_{i_2}\partial \theta_{i_1}}.
\end{align}
{With this, the multi derivative of the loss reads}
\begin{align}
   \partial^\alpha \mathcal{L}({\theta}) &= \bra{0} \partial^\alpha \left( {U}^{\dagger }({\theta} )OU({\theta} ) \right) \ket{0} \\
   &=  \sum_{\beta: \beta \leq \alpha}\left(\begin{array}{l}\alpha \\ \beta\end{array}\right)\bra{0}\left(\partial^{\beta}{U}^{\dagger }({\theta} )\right) O \left(\partial^{\alpha-\beta} U({\theta})\right) \ket{0},\nonumber 
\end{align}
where we have expolited the generalized Leibniz formula
\begin{align}
    \partial^{\alpha}(f g)=\sum_{\beta: \beta \leq \alpha}\left(\begin{array}{l}\alpha \\ \beta\end{array}\right)\left(\partial^{\beta} f\right)\partial^{\alpha-\beta} g.
\end{align}
{We have}
\begin{align}
    \vert \partial^\alpha \mathcal{L}\vert 
    &\le \sum_{\beta: \beta \leq \alpha}
    \left(\begin{array}{l}\alpha \nonumber\\ \beta\end{array}\right)\left|\bra{0}\left(\partial^{\beta}{U}^{\dagger }({\theta} )\right) O \left(\partial^{\alpha-\beta} U({\theta})\right) \ket{0}\right| \nonumber \\
    &= 2^{|\alpha|} \max_{\gamma: \gamma \leq \alpha}\left|\bra{0}\left(\partial^{\gamma}{U}^{\dagger }({\theta} )\right) O\, \partial^{\alpha-\gamma} U({\theta})\ket{0}\right| \nonumber \\
    &\le 2^{|\alpha|} \max_{\gamma: \gamma \leq \alpha} \left\Vert
    \left(\partial^{\gamma}{U}^{\dagger }({\theta} )\right) O\, \partial^{\alpha-\gamma} U({\theta})
    \right\Vert_\infty \nonumber \\
    &\le 2^{|\alpha|} \max_{\gamma: \gamma \leq \alpha}
    \left\Vert\partial^{\gamma}{U}^{\dagger }({\theta} ) \right\Vert_\infty
    \Vert O\Vert_\infty
    \left\Vert\partial^{\alpha-\gamma} U({\theta}) \right\Vert_\infty,
    \label{kderivativescomp}
\end{align}
where we have used the triangle inequality and 
the multi-binomial theorem formula to write 
\begin{equation}
\sum_{\beta: \beta \leq \alpha}\left(\begin{array}{l}\alpha \\ \beta\end{array}\right)=2^{|\alpha|}, 
\end{equation}
the fact that 
\begin{equation}
\left|\bra{0}A\ket{0}\right| \le \abs{A\ket{0}}_2 \le \abs{A}_\infty, 
\end{equation}
{which follows immediately by Cauchy-Schwarz and the sub-additivity of the $\Vert\cdot\Vert_\infty$ norm.}
Using the form of the parameterised unitary in Eq.~\eqref{eq:circuit_structure}, we can also observe that
\begin{align}
    \Vert \partial^{\gamma}{U}^{\dagger }({\theta} )\Vert_\infty
    &= \left\Vert
        \frac{\partial^{\gamma_p}}{\partial \theta^{\gamma_p}_p}\cdots\frac{\partial^{\gamma_1}}{\partial \theta^{\gamma_1}_1}{U}^{\dagger }({\theta})
    \right\Vert_\infty \nonumber \\
    &\le \Vert X_1\Vert^{\gamma_1}_\infty \cdots \Vert X_p\Vert^{\gamma_p}_\infty\nonumber \\ 
    &\le\left(\max_{j=1,\dots,p}\Vert X_j\Vert_\infty\right)^{\gamma_1+\cdots+\gamma_p} \nonumber \\
    &\le \max_{j=1,\dots,p} \Vert X_j\Vert_{\infty}^{\left|\gamma\right|}
\end{align}
where we have used that the sub-additivity of the infinity norm and the fact that the spectral norm of a unitary matrix is given by the unity.
Similarly, we have
\begin{equation}
    \Vert \partial^{\alpha-\gamma}{U}({\theta}) \Vert_\infty
    \le \max_{j=1,\dots,p} \Vert X_j\Vert_{\infty}^{|\alpha|-|\gamma|}.
\end{equation}
Therefore, 
combining the previous two inequalities with Eq.~\eqref{kderivativescomp}, we have
\begin{align}
    | \partial^\alpha \mathcal{L} |
    &\le 2^{|\alpha|} \Vert O\Vert_\infty \max_{j=1,\dots,p}\Vert X_j\Vert_{\infty}^{|\alpha|} \nonumber \\
    &= 2^k \Vert O\Vert_\infty \max_{j=1,\dots,p}\Vert X_j\Vert_{\infty}^k,
\end{align}
where we have used $\left|\alpha\right|=k$.
\end{proof}

{We are now ready to provide the proof of Theorem 
14 of the main text.
Since the loss function is a combination of sin and cosine functions, its derivatives exist and are bounded, and from this it follows that the loss function is strongly smooth and its Hessian is Lipschitz.
However, it is worth explicitly calculating $\beta$ and $\rho$ and bounding them to verify, for example, the scaling with the number of qubits.}
\begin{proof}
{We have the $\beta$-smooth constant defined by
the smallest $\beta$ with}
\begin{align}
    \abs{\partial \mathcal{L}({\theta})-\partial \mathcal{L}( {\theta^\prime})}\le \beta\abs{{\theta}- {\theta^\prime}},
\end{align}
which means we need to consider the Lipschitz constant for the $p$-dimensional function $\partial \mathcal{L}({\theta})$.
Using Lemma $\ref{LemmaRPtoRMgeneral}$ where $g({\theta})=\partial \mathcal{L}$ and $M=p$, we have
\begin{equation}
    \beta\le p\max_{i,j}\left(\sup_{{\theta}}\left|\frac{\partial^2 \mathcal{L}({\theta})}{\partial \theta_i\partial \theta_j}\right|\right).
\end{equation}
Applying Lemma \ref{SupDeriv}, 
we find
\begin{equation}
    \beta\le 2^2p \Vert O\Vert_\infty \max_{i}(\Vert X_i\Vert_\infty)^2,
\end{equation}
where we have used the matrix spectral (operator) norm.

The $\rho$-Hessian constant
is defined as
\begin{align}
    \Vert \partial^2 \mathcal{L}({\theta})-\partial^2 \mathcal{L}( {\theta^\prime}) \Vert_{\operatorname{H.S.}}
    \le \rho \Vert \theta- \theta^\prime \Vert_2
\end{align}
where $\partial^2 \mathcal{L}$ is the Hessian matrix and we have used the Hilbert-Schmidt matrix norm. 
{Note that the Hilbert Schmidt norm of a matrix is the 2-norm of the matrix \emph{vectorization} $\operatorname{vec}(\cdot)$.}
{We can now apply Lemma $\ref{LemmaRPtoRMgeneral}$, where $M=p^2$ since the Hessian is a map from $\mathbb{R}^p$ to $\mathbb{R}^{p\times p}$.
Defining $g({\theta})=\operatorname{vec}\left(\partial^2 \mathcal{L}({\theta})\right)$ we find}
\begin{equation}
    \rho\le p^\frac{3}{2}\max_{i,j,k}\left(\sup_{{\theta}}\left|\frac{\partial^3 \mathcal{L}({\theta})}{\partial \theta_k \partial \theta_i\partial \theta_j}\right|\right) .
\end{equation}
Thus, applying Lemma~\ref{SupDeriv}, we 
\je{arrive at}
\begin{equation}
    \rho \le 2^3p^\frac{3}{2} \Vert O\Vert_\infty\max_{i}(\Vert X_i\Vert_\infty)^3.
\end{equation}
\end{proof}

\subsection{Discussion on more general noise}
\label{sec:general-noise}
In this section, we discuss the impact of more general noise.
We assume that we have device noise that is constant in time, i.e., for each state preparation, effectively\je{,} we always 
encounter the same CPTP maps acting on the initial state.
This will change the state $\rho(\theta)$ 
at the end of the circuit to a noisy 
instance
$\rho_{\text{noisy}}(\theta)$ which can be modelled as the applications of parametrised unitary layers interspersed with 
\je{suitable} 
CPTP maps to {the initial state $\rho_0$ as
\begin{align}
    \rho_{\text{noisy}}(\theta)=\mathcal{N}_p\circ\mathcal{U}_p\circ \cdots \circ \mathcal{N}_1\circ\mathcal{U}_1\left(\rho_0\right)
    \label{eq:noisystate}
\end{align}
{where $\mathcal{N}_i$ and $\mathcal{U}_i$ are respectively noisy CPTP maps and the parametrised unitary channels.}
As a result, the loss function changes to
\begin{align}
    \mathcal{L}_{\text{noisy}}(\theta)=\Tr(H\rho_{\text{noisy}}(\theta)),
    \label{eq:noisycost}
\end{align}
i.e.\je{,} 
we now have to optimise an inherently different loss function.
Still, to estimate the noisy loss function 
$\mathcal{L}_{\text{noisy}}$\je{,} also here we will have to deal with statistical noise derived by the finite number of measurements.
In particular for the case of global depolarising noise with depolarising noise parameter $q\in \left[0,1\right]$
\begin{align}
    \mathcal{N}_i\left(\cdot\right)=\left(1-q\right)\left(\cdot\right)+q\Tr\left(\cdot\right)\frac{\mathbb{1}}{2^n}
\end{align} the cost function will be
\begin{align}
     \mathcal{L}_{\text{noisy}}(\theta)=\left(1-q\right)^p \mathcal{L}\left(\theta\right) + \left(1-\left(1-q\right)^p\right)\frac{\Tr\left(H\right)}{2^n}.
\end{align}
Therefore, in this case, the landscape of the cost function will be rescaled and shifted, but will preserve features of the noiseless-landscape like the position of saddle points.
\begin{proof}
We have 
\begin{align*}
    \rho_{\text{noisy}}(\theta)=&\left(1-q\right)\mathcal{N}_p\circ\mathcal{U}_p\circ \cdots \circ \mathcal{N}_2\circ\mathcal{U}_2\left(\mathcal{U}_1\left(\rho_0\right)\right)\\&+ q\,  \frac{\mathbb{1}}{2^n}\\
    =&\left(1-q\right)^2\mathcal{N}_p\circ\mathcal{U}_p\circ \cdots \circ \mathcal{N}_3\circ\mathcal{U}_3\left(\mathcal{U}_2\circ\mathcal{U}_1\left(\rho_0\right)\right)\\&+ q\left(\left(1-q\right)+ 1 \right) \left(\frac{\mathbb{1}}{2^n}\right)\\
    =&\left(1-q\right)^p\rho\left(\theta\right)+q\left(\sum^{p-1}_{k=0} (1-q)^k\right)\frac{\mathbb{1}}{2^n}\\
    =&\left(1-q\right)^p\rho\left(\theta\right)+q\frac{1-\left(1-q\right)^p}{q}\frac{\mathbb{1}}{2^n}.
\end{align*}
Plugging such a state into the definition of the noisy cost function (\ref{eq:noisycost}), we have the result.
\end{proof}}
\subsection{Additional numerical results}
\je{In this section,} we show additional numerical results to the one reported in the main text. 
In Fig.~\ref{fig:performanceShotsQML}\je{,} we show the estimated probability of avoiding the saddle point as a function of the number of shots, for the loss function given by the expectation value of the
\je{local Pauli} Hamiltonian ${H}=\sum_{i=1}^{N=4}Z_i$ over the circuit $\texttt{qml.StronglyEntanglingLayers}$ (see Fig.~\ref{fig:strongly-entangling-layer})
in \emph{Pennylane}~\cite{bergholm2018pennylane} where two layers 
\je{have been}
used. 
\begin{figure}[h]
\centering
\includegraphics[width=0.5 \textwidth]{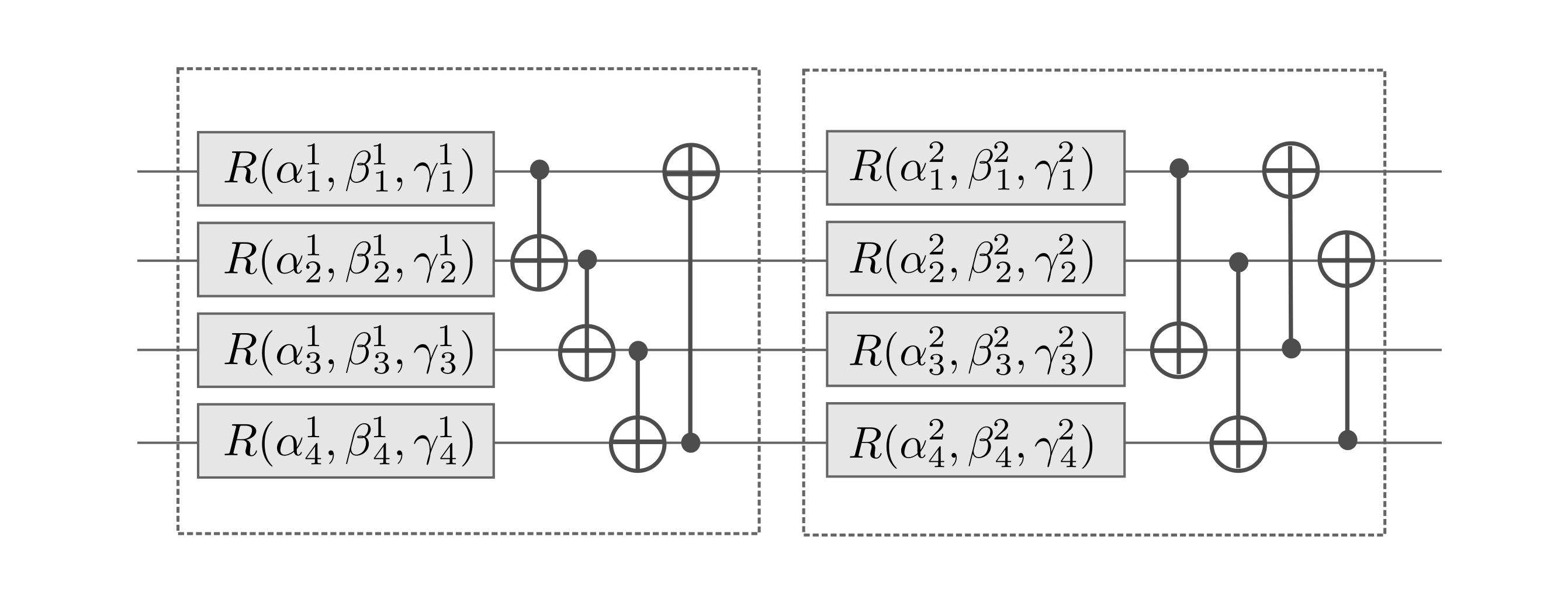}

\caption{{A four-qubit example of one strongly entangling layer as given in $\texttt{qml.StronglyEntanglingLayers}$ in Pennylane. The figure has been adopted from the documentation of Pennylane~\cite{bergholm2018pennylane}.}} %
\label{fig:strongly-entangling-layer}
\end{figure}

\begin{figure}[h]
\centering
\includegraphics[width=0.43\textwidth]{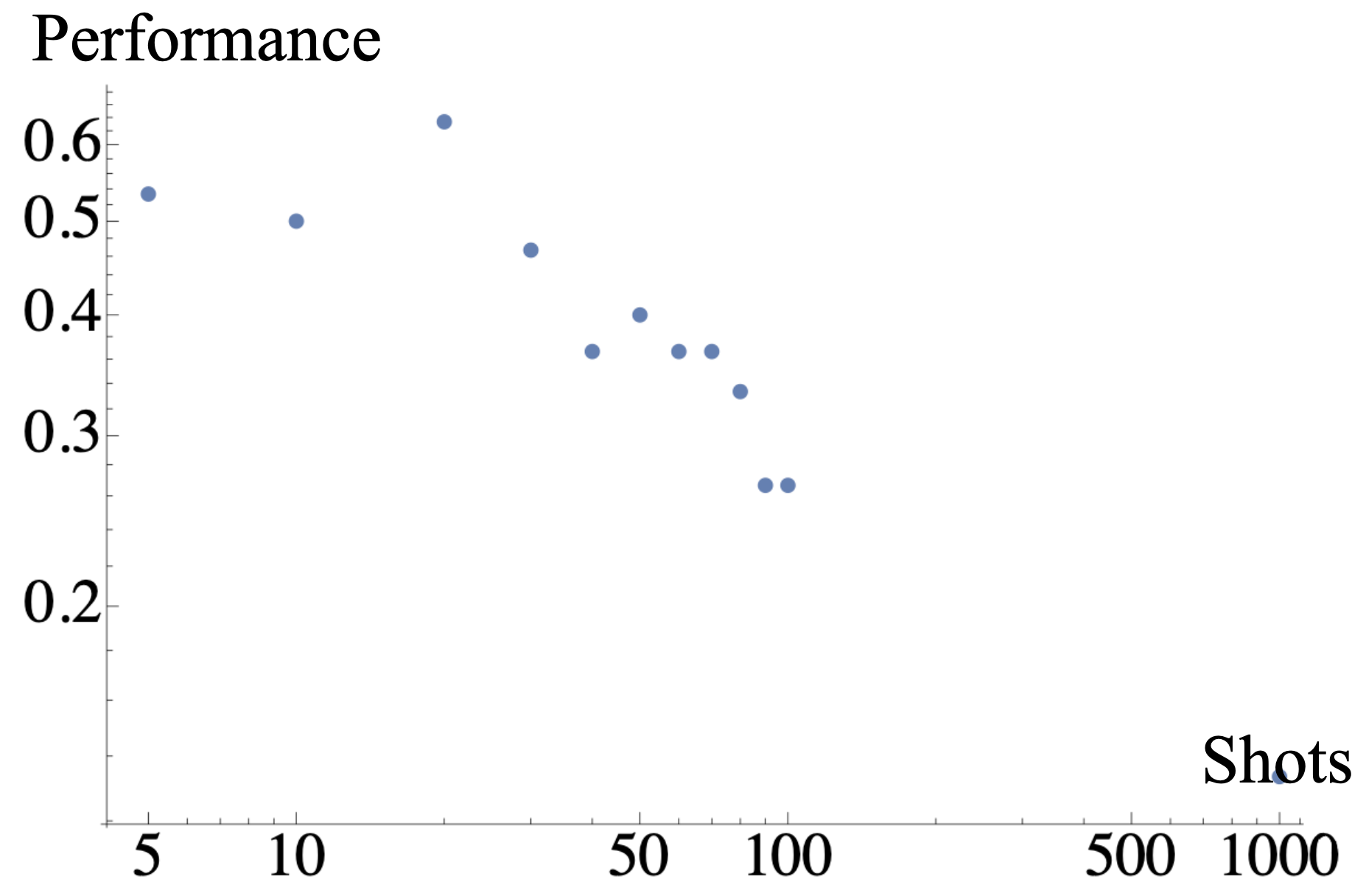}

\caption{
We quantify the performance against the number of shots of the quantum noise by the probability of saddle-point avoidance for 100 different independent instances with the same initial conditions.} %
\label{fig:performanceShotsQML}
\end{figure}

\begin{figure}[h]
\centering
\includegraphics[width=0.44\textwidth]{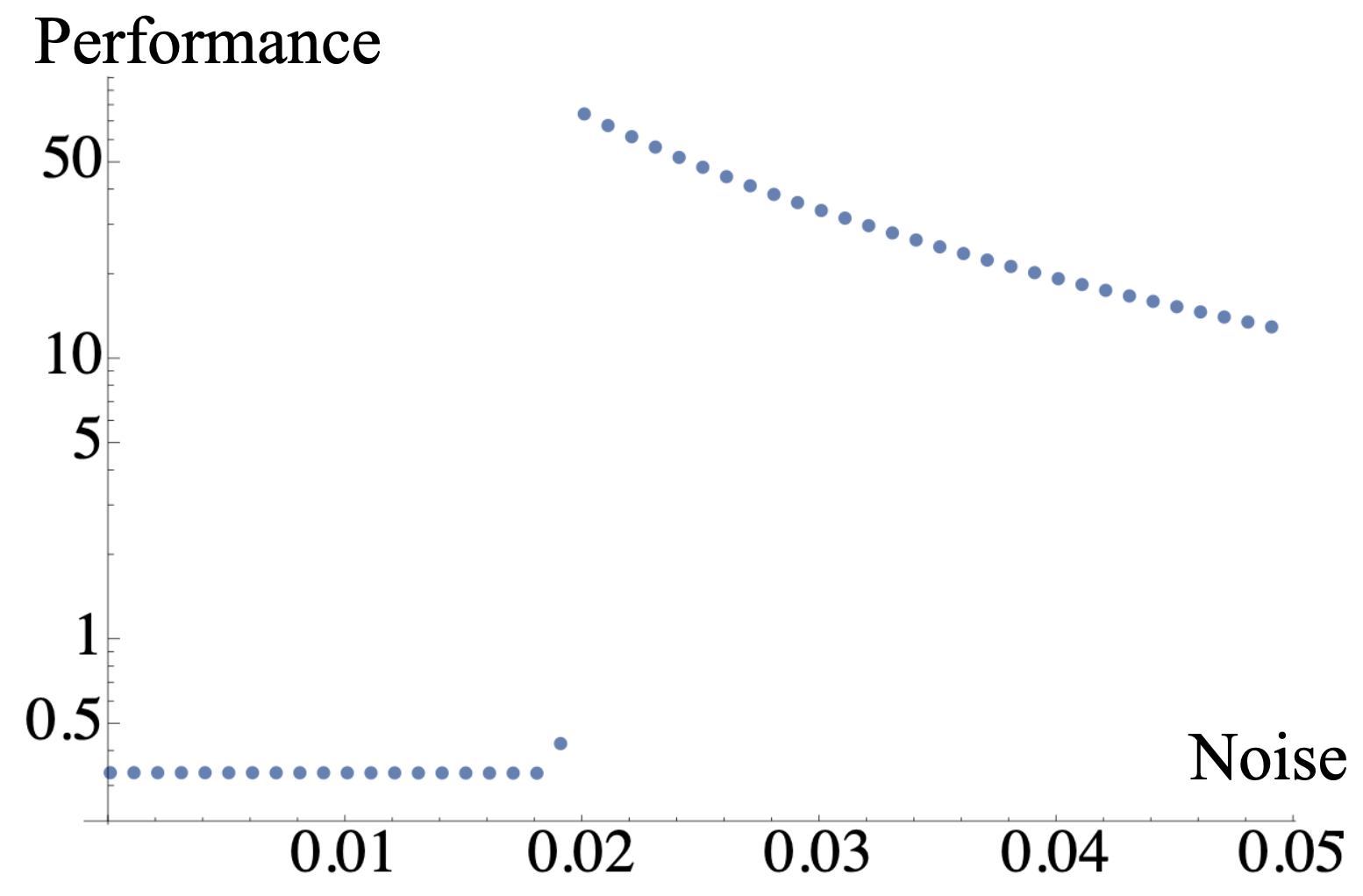}
\caption{We quantify the performance against the size of the noise $r$ (classical Gaussian noise) by ${1}/({\mathcal{L}-\mathcal{L}_\text{opt}})$. We 
\je{again have eight}
qubits.}
\label{fig:performance_morequbits}
\end{figure}

One could directly extend our results towards \je{the description of situations involving} more qubits. 
\je{Here, we} consider 
\je{eight}
qubits and two layers
\je{of quantum gates}. Now the loss landscape is 
\je{richer}
and we can converge at more integers. For instance, we find convergence at $-3$, $-4$, $-5$, $-6$ for different initial conditions. 
Figure \ref{fig:fake_more} illustrates saddle-point avoidance with different noise levels when noise is selected from Gaussian distributions. Figure \ref{fig:performance_morequbits} illustrates the performances among different sizes of noise levels and one can again find a critical value of the noise which leads to the saddle-point avoidance. 
\begin{figure}[h]
\centering
\includegraphics[width=0.45\textwidth]{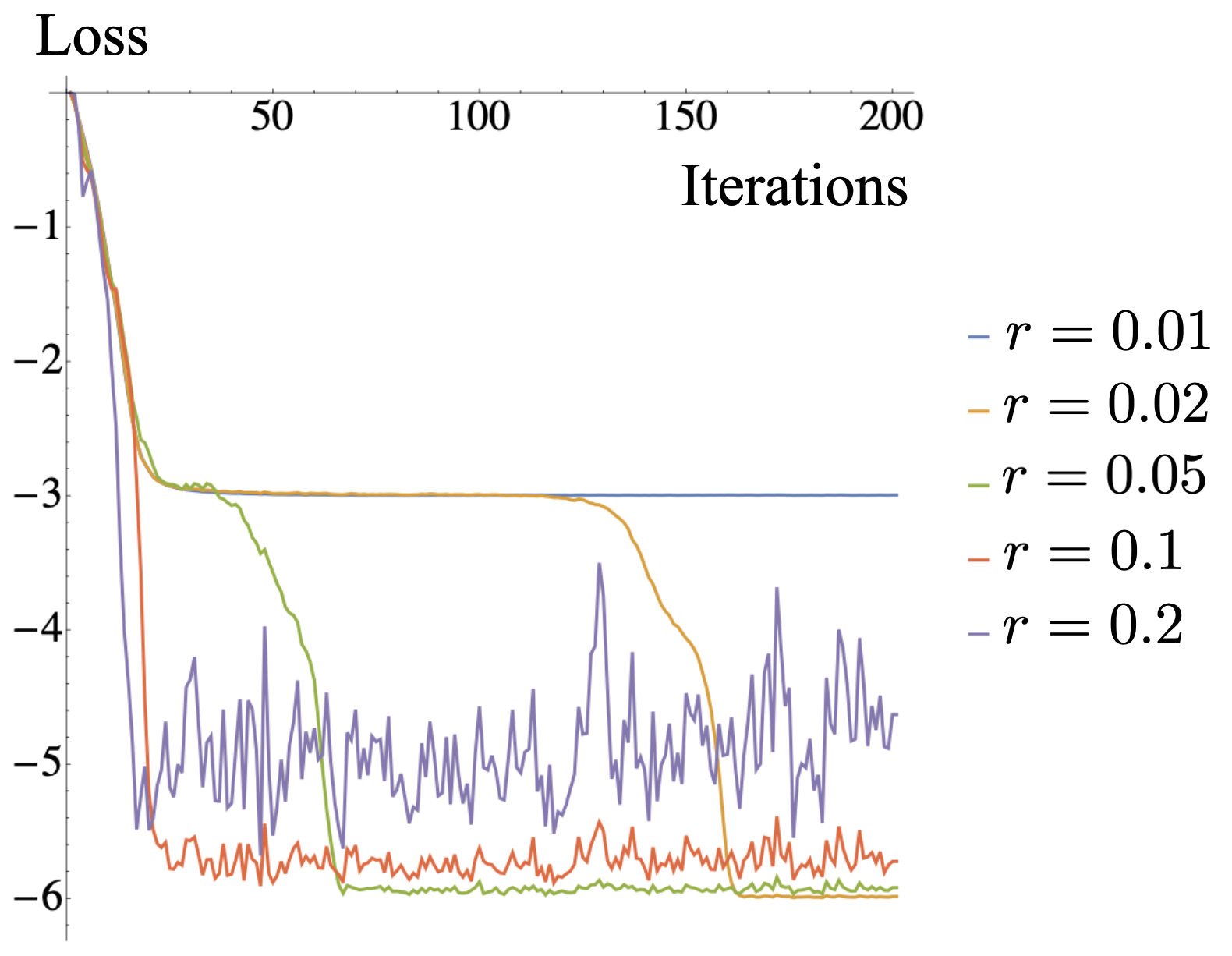}
\caption{Comparison of the loss evolution with or without noise with 
\je{eight} qubits. The noise has been drawn
manually from Gaussian distributions, and we keep the same initial conditions. We use four different values of the noise norms.}
\label{fig:fake_more}
\end{figure}
In Fig.~\ref{fig:performance_vqeH2}\je{,} we 
\je{depict}
the performances as a function of the noise level obtained for the $\text{H}_2$ molecule experiment. 

\begin{figure}[h]
\centering
\includegraphics[width=0.41\textwidth]{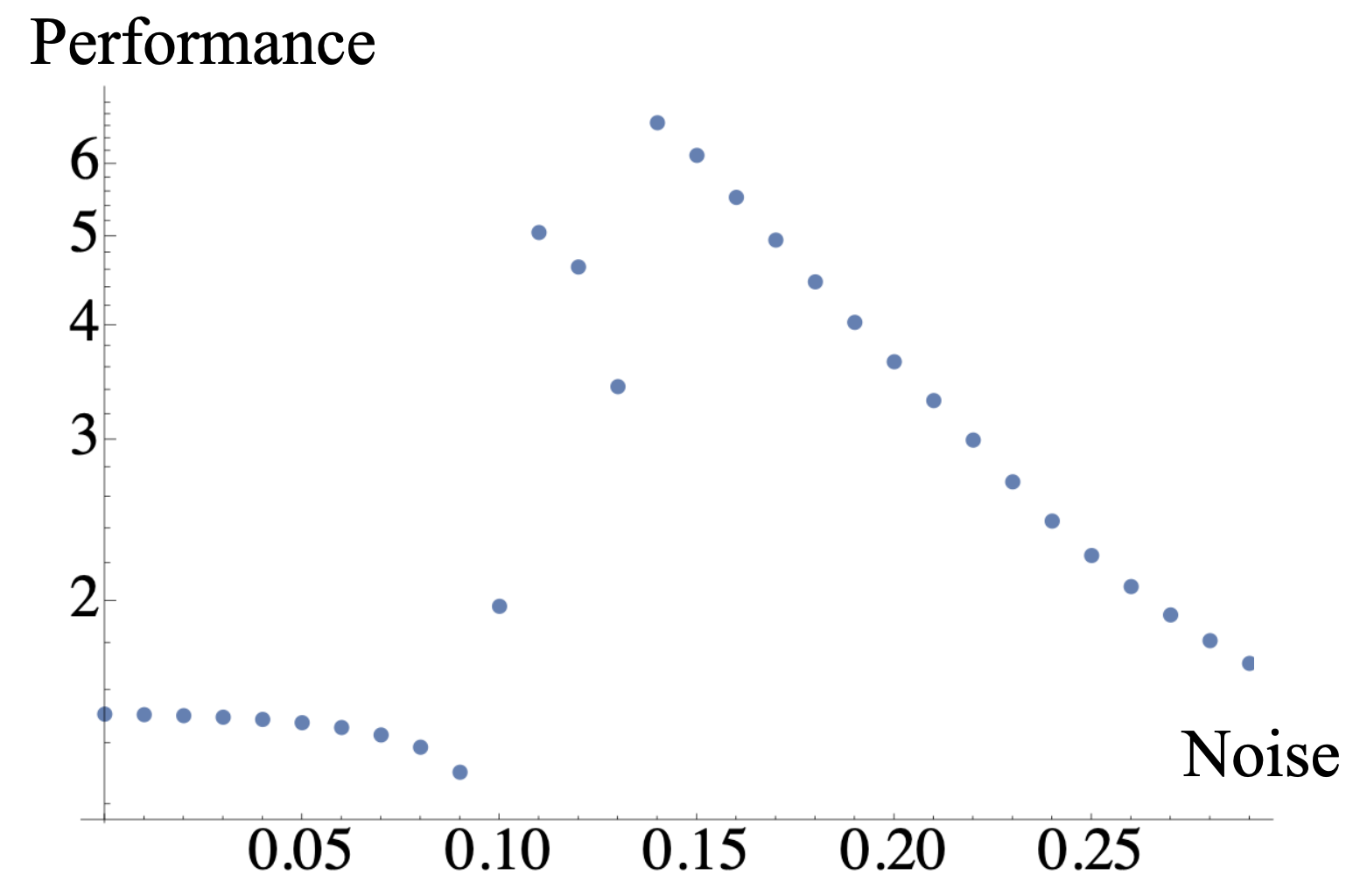}
\caption{In the Hydrogen VQE example, we quantify the performance against the size of the noise $r$ (classical Gaussian noise) by ${1}/({\mathcal{L}-\mathcal{L}_\text{opt}})$. }
\label{fig:performance_vqeH2}
\end{figure}

Furthermore, we try to find the relation between the convergence time $T$ and the noise size $r\sim \epsilon$. With the same setup, we plot the dependence between the convergence time (the time where we approximately get the true minimum) and the size of the noise in the Gaussian distribution case, in Fig.~\ref{fig:fittingt}. We find that the convergence time indeed decays when we add more noise, and we fit the scaling and find where $T\sim \#/\epsilon^{0.6}$, which is consistent with the bound $T\sim \#/\epsilon^2$ in theory. In Section \ref{sec:analyticheu}, we provide a heuristic derivation on the scaling $T\sim 1/\epsilon^2$ by dimensional analysis and other analytic heuristics. 

\begin{figure}[h]
\centering
\includegraphics[width=0.41\textwidth]{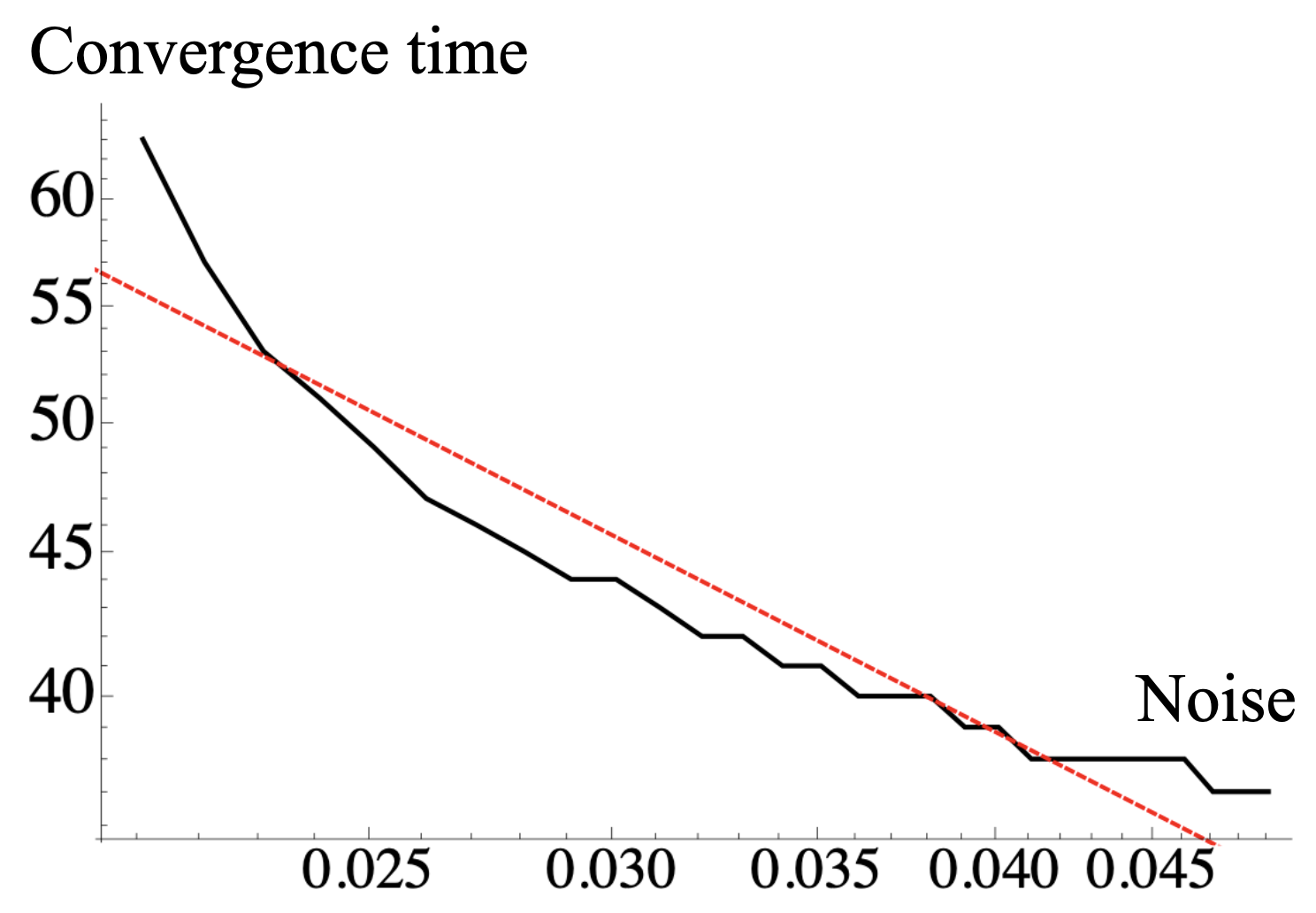}
\caption{The 
\je{relationship}
between the convergence time $T$ and the size of the noise $r \sim \epsilon$. The data is plotted in black, and we fit the data using $\#/\epsilon^\Delta$ in red, and we get $\Delta \approx 0.6 $. The setup is the same as before: We use 
\je{four} qubits and 
\je{two} layers for our first example Hamiltonian, and we use Gaussian noise simulation.
}
\label{fig:fittingt}
\end{figure}

\subsection{Analytic heuristics}\label{sec:analyticheu}
In this section, we provide a set of analytic heuristics about predicting the noisy convergence and the critical noise with significant improvements in performance. Our derivation is physical and heuristic, but we expect that they will be helpful to understand the nature of the noisy dynamics during gradient descent in the quantum devices. The results developed corroborate the idea that a balance between too little and too much noise will have to be struck.

\subsubsection{Brownian motion and the Polya's constant}
One of the simplest heuristics about noisy gradient descent is the theory of Brownian motion. 
Define $p(d)$, also known as Polya's constant, as the likelihood that a random walk on a $d$-dimensional lattice has the capability to return to its starting point.
It has been proven that \cite{polya1921aufgabe},
\begin{align}
    p(1)=p(2)=1~,
\end{align}
but 
\begin{align}
   p(d\ge 3) <1 ~.
\end{align}
\je{In fact, $d\mapsto p(d)$
has the closed
formula \cite{montroll1956random}}
\begin{align}
p(d)=1-{{\left( \int_{0}^{\infty }{{{\left[ {{I}_{0}}\left( \frac{t}{d} \right) \right]}^{d}}}{{e}^{-t}}dt \right)}^{-1}}~,
\end{align}
where $d>3$ is the number of training parameter in our case, and $I$ is the modified Bessel function of the first kind. One could compute numerical values of the probability $p(d)$ for increasing $d$. From $d=4$ to $d=8$, it changes monotonically from 0.19 to 0.07. It is hard to compute the integral accurately because of damping, but it is clear that it is decaying and will vanish for large $d$. 
In our problem, we could regard the process of noisy gradient descent as random walks in the space of variational angles. One could regard the returning probability roughly as the probability of coming back to the saddle point from the minimum. Thus, the statement about lattice random walk gives us intuition that it is less likely to return back when we have a large number of variational angles.

\subsubsection{Guessing $1/\epsilon^2$ by dimensional analysis}
One of the primary progress of the technical result
\je{presented in Ref.}~\cite{ExpTimeSaddle} is the $1/\epsilon^2$ dependence on the convergence time $T$ with the size of the noise \je{$\epsilon>0$}. 
Here\je{,} we show that one could guess such a result in the small $\eta$ limit (where $\eta$ is the learning rate) simply by dimensional analysis. 
Starting from the definition of the gradient descent algorithm,
\begin{align}
\delta {\theta _i } = {\theta _i }(t + 1) - {\theta _i }(t) =  - \eta \frac{{\partial \mathcal{L}}}{{\partial {\theta _i }}}~,
\end{align}
we 
\je{can}
instead study the variation of the loss function
\begin{align}
&\delta \mathcal{L} = \mathcal{L}(t + 1) - \mathcal{L}(t) \approx \sum\limits_i  {\frac{{\partial \mathcal{L}}}{{\partial {\theta _i }}}\delta {\theta _i }}  =  - \eta \sum\limits_i  {\frac{{\partial \mathcal{L}}}{{\partial {\theta _i }}}\frac{{\partial \mathcal{L}}}{{\partial {\theta _i }}}} \nonumber\\
&=  - 4\eta \sum\limits_i  {\frac{{\partial \sqrt \mathcal{L} }}{{\partial {\theta _i }}}\frac{{\partial \sqrt \mathcal{L} }}{{\partial {\theta _i }}}} \mathcal{L} ~.
\end{align}
Here, we use the assumption where $\eta$ is small, such that we could expand the loss function change $\delta \mathcal{L}$ by the first order Taylor expansion. Now we define
\begin{align}
{K_{\mathcal{L}}} := 4\sum\limits_i  {\frac{{\partial \sqrt \mathcal{L} }}{{\partial {\theta _i }}}\frac{{\partial \sqrt \mathcal{L} }}{{\partial {\theta _i }}}} ~,
\end{align}
and we have
\begin{align}
\delta \mathcal{L} =  - \eta {K_\mathcal{L}}\mathcal{L} ~.
\end{align}
If $K_\mathcal{L}$ is a constant (and we could assume it is true since we are doing dimensional analysis), we get
\begin{align}
\mathcal{L}(t) = {(1 - \eta {K_\mathcal{L}})^t} \approx {e^{ - \eta {K_\mathcal{L}}t}}~.
\end{align}
In general, we can assume a time-dependent solution as
\begin{align}
\mathcal{L}(t) = {(1 - \eta {K_\mathcal{L} (t)})^t} \approx {e^{ - \eta {K_\mathcal{L} (t)}t}} ~.
\end{align}
Now let us think about how the scaling of convergence time will be with noise. First, in the $\eta \to 0$ limit, for small $\eta$ the convergence time would get smaller, not larger (since it is immediately dominated by noise). So it is not possible that $T\sim 1/\eta $ to some powers in $\eta$. For this reason, {the only possibility left is} 
\begin{align}
T = \mathcal{O}(1) + \mathcal{O}(\eta ) + \mathcal{O}({\eta ^2}) \ldots ~
\end{align}
in the scaling of $\eta$. We will thus focus on the first $\mathcal{O}(1)$ term in the small $\eta$ limit. Furthermore, from the form $e^{ - \eta {K_{\mathcal{L}}}t}$ we know that $T\sim 1/K_{\mathcal{L}}$. 

Now let us count the dimension, assuming $\theta_i$ has the $\theta$-dimension 1 and $L$ has the $\theta$-dimension 0. From the gradient descent formula, $\eta$ has the $\theta$-dimension $2$, $K_{\mathcal{L}}$ has the $\theta$-dimension $-2$, and $\epsilon$ has the $\theta$-dimension $1$. The time $T$ is dimensionless since $\eta {K_{\mathcal{L}}}T$ is dimensionless and 
\je{appears}
in the exponent. Thus, since we know that $T\sim 1/K_{\mathcal{L}}$, there must be an extra factor balancing the $\theta$-dimension of $K_{\mathcal{L}}$. The only choice is $\epsilon^2$, and we cannot use $\eta$ because we are studying the term with the $\eta$-scaling $\mathcal{O}(1)$. Thus, we immediately get, $T\sim 
{1}/{(K_{\mathcal{L}} \epsilon^2)}$. That is how we get the dependence $T\sim 1/\epsilon^2$ by dimensional analysis. Note that the estimation only works in the small $\eta$ limit. More generally, we have
\begin{align}
    T=\sum\limits_{m,n>2m}^{{}}{\mathcal{O}
    \left(\frac{{{\eta }^{n-2m}}}{K_{\mathcal{L}}^{m}{{\epsilon }^{n}}}\right)}~
\end{align}
if we assume that the expression of $T$ is analytic. 

\subsubsection{Large-width limit}
The dependence $T\sim 1/\epsilon^2$ can also be made plausible using the \emph{quantum neural tangent kernel} (QNTK) theory. The QNTK theory has been established \cite{Liu:2021wqr,Liu:2022eqa,Liu:2022rhw} in the limit where we have a large number of trainable angles $d$ and a small learning rate $\eta$, with the quadratic loss function. According to Ref.~\cite{Liu:2022rhw},
we use the loss function
\begin{align}
\mathcal{L}(\theta)=\frac{1}{2}\left(\left\langle\Psi_{0}\left|U^{\dagger}(\theta) O U(\theta)\right| \Psi_{0}\right\rangle-O_{0}\right)^{2} =: \frac{1}{2} \varepsilon^{2} ~.
\end{align}
Here, we make predictions on the eigenvalue of the operator $O$ towards $O_0$. And we use $U(\theta) $ as the variational ansatz. The gradient descent algorithm is
\begin{align}
{{\theta }_{i}}(t+1)-{{\theta }_{i}}(t)=: 
\delta {{\theta }_{i}}=-\eta \frac{\partial{\mathcal{L}}}{\partial {{\theta }_{i}}}~
\end{align}
when there is no noise. Furthermore, we hereby model the noise by adding Gaussian random variables in each step of the update. Those random fluctuations are independently distributed through $\Delta \theta_i\sim \mathcal{N} (0,\epsilon^2)$. 
Now, in the limit where $d$ is large, we have an analytic solution of the convergence time,
given by
\begin{align}
T\approx \frac{\log \left( \frac{\epsilon }{\sqrt{2{{\varepsilon }^{2}}(0)\eta -{{\varepsilon }^{2}}(0){{\eta }^{2}}K+{{\epsilon }^{2}}}} \right)}{\log (1-\eta K)}~,
\end{align}
where $K:=K_{\mathcal{L}}/2$. In the small $\eta$ limit, we have
\begin{align}
T\approx \frac{{{\varepsilon }^{2}}(0)}{\epsilon^2 K} ~.
\end{align}
This gives substance to the claim in 
the dimensional analysis.

\subsubsection{Critical noise {from random walks}}
Moreover, using the result from Ref.~\cite{Liu:2022rhw}, we 
\je{can}
also estimate the critical noise $\epsilon_{\text{cri}}$, namely, the critical value of phase transition of the noise size that leads to better performance and avoiding the saddle points. 

{In particular, here we will be interested in the case where the saddle-point avoidance is triggered purely by random walks without any extra potential. The assumption, although may not be real in the practical loss function landscape, might still provide some useful guidance.}
According to Ref.~\cite{Liu:2022rhw}, we have
\begin{align}
\overline{{{\varepsilon }^{2}}}(t)={{(1-\eta K)}^{2t}}\left( {{\varepsilon }^{2}}(0)-\frac{\epsilon^{2}}{\eta (2-\eta K)} \right)+\frac{\epsilon^2}{\eta (2-\eta K)}~.
\end{align}
Here $\overline{\varepsilon^2}$ is the variance of the residual training error $\varepsilon$ after averaging over the realizations of the noise. Imagine that now the gradient descent process is running from the saddle point to the exact local minimum, we have
\begin{align}
\frac{1}{2}\left( {{\left| {{\varepsilon }_{\text{saddle}}} \right|}^{2}}-{{\left| {{\varepsilon }_{\text{min}}} \right|}^{2}} \right)=\Delta_{\mathcal{L}}\sim \frac{\epsilon ^{2}}{2\eta (2-\eta K)}~,
\end{align}
where $\Delta_{\mathcal{L}}$ is the distance of the loss function from the saddle point to the local minimum (defined also in the main text), $\Delta_{\mathcal{L}} = \mathcal{L}_{\text{saddle}}-\mathcal{L}_{\text{mininum}}=\frac{1}{2}\left( {{\left| {{\varepsilon }_{\text{saddle}}} \right|}^{2}}-{{\left| {{\varepsilon }_{\text{min}}} \right|}^{2}} \right)$. So we get an estimate of the critical noise,
\begin{align}
\epsilon_{\text{cri}}^2\sim\Delta_{\mathcal{L}}\left( 2\eta (2-\eta K) \right)\sim 4\eta \Delta_{\mathcal{L}}~.
\end{align}
Here\je{,} in the most right hand side of the formula\je{,} we use the approximation where $\eta$ is small enough. This formula might be more generic beyond QNTK, since one could regard it as an analog of Einstein's formula of \emph{Brownian motion},
\begin{align}
\overline{x^2}(t)= 2 D t ~,
\end{align}
with the averaging moving distance square $\overline{x^2}$, mass diffusivity $D$, and time $t$ in the Brownian motion.

{One can also show such a scaling in the linear model. Say that we have a linear loss function
\begin{align}
\mathcal{L} = \sum\limits_\mu  {{c_\mu }{\theta _\mu }}  + b~,
\end{align}
with constants $c_\mu$ and $b$. For simplicity, we assume that the initialisation $\theta(0)$ makes $\mathcal{L}(\theta(0))=\mathcal{L}(0)>0$. The gradient descent relation is
\begin{align}
\delta {\theta _\mu } = {\theta _\mu }(t + 1) - {\theta _\mu }(t) =  - \eta \frac{{\partial \mathcal{L} }}{{\partial {\theta _\mu }}} =  - \eta {c_\mu }~.
\end{align}
One can find the closed-form solution
\begin{align}
{\theta _\mu }(t) = {\theta _\mu }(0) - \eta t{c_\mu }~.
\end{align}
It is also possible to identify the change of the loss function to be
\begin{align}
\mathcal{L} (t) = \sum\limits_\mu  {{c_\mu }{\theta _\mu }(0)}  + b - \eta t\sum\limits_\mu  {c_\mu ^2}  = \mathcal{L} (0) - \eta t\sum\limits_\mu  {c_\mu ^2} ~.
\end{align}
The convergence time can be estimated as
\begin{align}
T = \frac{{\mathcal{L} (0)}}{{\eta \sum\limits_\mu  {c_\mu ^2} }}~.
\end{align}
Now, instead, we add a random $\xi_\mu (t)$ in the gradient descent dynamics, which is following the normal distribution ${\xi _\mu }(t) \sim \mathcal{N}(0,\sigma_{\mu}^2)$. Now, the stochastic gradient descent equation is
\begin{align}
\delta {\theta _\mu } = {\theta _\mu }(t + 1) - {\theta _\mu }(t) =  - \eta \frac{{\partial \mathcal{L} }}{{\partial {\theta _\mu }}} + {\xi _\mu } =  - \eta {c_\mu } + {\xi _\mu }~,
\end{align}
which gives the solution
\begin{align}
{\theta _\mu }(t) = {\theta _\mu }(0) - \eta t{c_\mu } + \sum\limits_{i = 0}^{t - 1} {{\xi _\mu }(i)} ~.
\end{align}
Thus, we get the loss function
\begin{align}
&\mathcal{L}(t) = \sum\limits_\mu  {{c_\mu }{\theta _\mu }(0)}  + b - \eta t\sum\limits_\mu  {c_\mu ^2}  + \sum\limits_{\mu ,i = 0}^{t - 1} {{c_\mu }{\xi _\mu }(i)} \nonumber\\
&= \mathcal{L}(0) - \eta t\sum\limits_\mu  {c_\mu ^2}  + \sum\limits_{\mu ,i = 0}^{t - 1} {{c_\mu }{\xi _\mu }(i)} ~.
\end{align}
The critical point ${\sigma _\mu } = {\varepsilon _{{\rm{cri}}}}$ can be identified as
\begin{align}
\eta t\sum\limits_\mu  {c_\mu ^2}  \sim {\epsilon _{{\rm{cri}}}}\sqrt t 
\left( {\sum\limits_\mu  {c_\mu ^2} }\right)^{1/2} ~,
\end{align}
where the standard deviation of the noise term will compensate the decay. Thus, we get
\begin{align}
{\epsilon _{{\rm{cri}}}} \sim \eta \sqrt t 
\left({\sum\limits_\mu  {c_\mu ^2} } 
\right)^{1/2}~.
\end{align}
In the limit where the noise
levels are small, we can study the behaviour in the late time
limit
\begin{align}
t=T=\frac{{\mathcal{L} (0)}}{{\eta \sum\limits_\mu  {c_\mu ^2} }}~.
\end{align}
So we get
\begin{align}
{\epsilon _{{\rm{cri}}}} \sim \sqrt {\eta \mathcal{L}(0)}  \sim \sqrt {\eta {\Delta _{\cal L}}} ~,
\end{align}
which is
\begin{align}
\epsilon _{{\rm{cri}}}^2 \sim \eta {\Delta _{\cal L}}~.
\end{align}
Thus, the linear model result is consistent with the derivation using QNTK with the quadratic loss. }

\subsubsection{{Phenomenological critical noise}}

{In practice, random walks may not be the only source triggering the saddle-point avoidance, leading to the $\sqrt{t}$ scaling in loss functions. Since saddle points have negative Hessian eigenvalues, those directions will provide driven forces with linear contributions $\propto t$ in the loss function. In the linear model, we can estimate the critical noise as}
\begin{align}
\eta t{\Delta _{\mathcal{L}}}\sim {\epsilon _{{\rm{cri}}}}t~,
\end{align}
{which leads to the linear relation}
\begin{align}
{\epsilon _{{\rm{cri}}}} \sim \eta {\Delta _{\mathcal{L}}}~.
\end{align}

{In Fig.~\ref{fig:critical}, we show the dependence of the critical noise ${\epsilon _{{\rm{cri}}}}$ on $\eta {\Delta _{\mathcal{L}}}$ in our numerical example with four qubits from Fig.~3 of the main text and its linear fitting. We find that our theory is justified for a decent range of learning rate. In our numerical data, in smaller learning rates the increase of the critical noise might be smoother, while for larger critical noise, the growth is more close to linear scaling. If we fit for the exponent of ${\epsilon _{{\rm{cri}}}} \sim {(\eta {\Delta _L})^{{\Delta _{{\rm{cri}}}}}}$ we get ${\Delta _{{\rm{cri}}}} \approx {\rm{0.8722}}$. }

{One could also transform critical learning rates towards the number of shots if the noises are dominated from quantum measurements. One can assume the scaling
\begin{align}
{\epsilon _{{\rm{cri}}}} \sim \frac{\eta }{{\sqrt {{N_{{\rm{cri}}}}} }}~,
\end{align}
and we can obtain the optimal number of shots
\begin{align}
{N_{{\rm{cri}}}} = {N_{{\rm{cri}}}} = {c_N}{\eta ^{2 - 2{\Delta _{{\rm{cri}}}}}}\Delta_{\mathcal{L}}^{ - 2{\Delta _{{\rm{cri}}}}}~.
\end{align}
One can then take $\Delta_{\text{cri}} = 1 $ as a good approximation, assuming that the saddle-point avoidance is dominated by negative saddle-point eigenvalues. ${c_{{\rm{cri}}}}$ is a constant depending on the circuit architecture and the loss function landscapes. This formalism could be useful to estimate the optimal number of shots used in variational quantum algorithms. For instance, we take $\Delta_{\text{cri}} = 1 $ and we get
\begin{align}
{N_{{\rm{cri}}}} = {c_N}\Delta_{\mathcal{L}}^{ - 2}~.
\end{align}
In the situation of Fig.~3 of the main text, we obtain
\begin{align}
\epsilon = {c_\eta}\frac{\eta }{{\sqrt N }}
\end{align}
with $c_\eta$ estimated as ${c_\eta} \approx {\rm{1.19733}}$ from \texttt{Qiskit}. So we get
\begin{align}
{N_{{\rm{cri}}}} = \frac{{c_\eta ^2}}{{c_{\epsilon}^2}}\frac{1}{{\Delta _{\cal L}^2}} = 131.8~,
\end{align}
which is the optimal numbers of shots in this experiment with pure measurement noises. Here, $c_N= {{c_\eta ^2}}/{{c_{\epsilon}^2}}$.}

\begin{figure}[h!]
\centering
\includegraphics[width=0.46\textwidth]{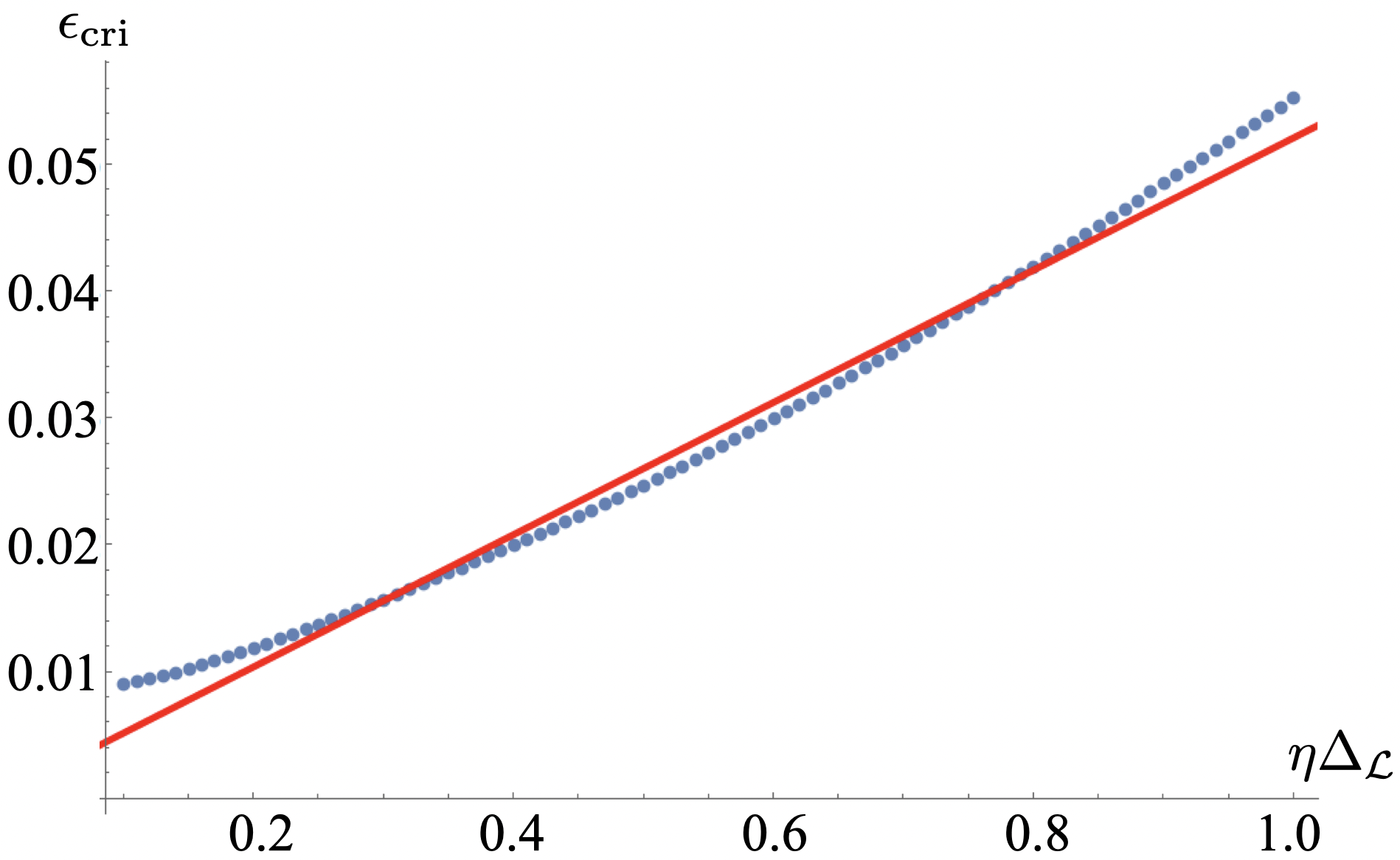}
\caption{{The dependence of the critical noise ${\epsilon _{{\rm{cri}}}}$ on $\eta {\Delta _{\mathcal{L}}}$ in the example of 
four qubits. Here, we fit the dependence by the linear relation ${\epsilon _{{\rm{cri}}}} = c_{\epsilon} \eta {\Delta _{\mathcal{L}}} $ where $c_{\epsilon} = 0.0521404$.} }
\label{fig:critical}
\end{figure}

\end{document}